\documentclass[alpha-refs]{wiley-article}

\usepackage{bbm,bm,enumitem,tikz,verbatim,multirow,multicol}
\usepackage{pgfplots}
\pgfplotsset{compat=1.16,every axis/.style={width=5cm}}
\usepgfplotslibrary{fillbetween}
\usepackage{color}
\usepackage{ulem}
\newtheorem{mylem}{Lemma}
\newtheorem{mycor}{Corollary}

\usepackage{siunitx}

\title{Consistent Covariance estimation for stratum imbalances under minimization method for covariate-adaptive randomization}

\author[1]{Zixuan Zhao}
\author[1]{Yanglei Song}
\author[1]{Wenyu Jiang}
\author[1,2]{Dongsheng Tu}

\affil[1]{Department of Mathematics and Statistics, Queen's University, Kingston, Ontario, K7L 3N8, Canada}
\affil[2]{Canadian Cancer Trials Group, Kingston, Ontario, K7L 2V5, Canada}

\corraddress{Dongsheng Tu, Canadian Cancer Trials Group, Kingston, Ontario, K7L 2V5, Canada.}
\corremail{dtu@ctg.queensu.ca}

\fundinginfo{Natural Sciences and Engineering Research Council of Canada (NSERC), Grant/Award Number: DGECR-2020-00337, NSERC RGPIN-2020-04256
}

\runningauthor{Z. Zhao et al.}

\begin{document}

\maketitle

\begin{abstract}

Pocock and Simon's minimization method is a popular approach for covariate-adaptive randomization in clinical trials. Valid statistical inference with data collected under the minimization method requires the knowledge of the limiting covariance matrix of within-stratum imbalances, whose existence is only recently established. In this work, we propose a bootstrap-based estimator for this limit and establish its consistency, in particular, by Le Cam's third lemma.  As an application, we consider in simulation studies adjustments to existing robust tests for treatment effects with survival data by the proposed estimator. It shows that the adjusted tests achieve a size close to the nominal level, and unlike other designs, the robust tests without adjustment may have an asymptotic size inflation issue under the minimization method. 
\keywords{covariance matrix estimation, covariate-adaptive design, minimization method, parametric bootstrap, size inflation, within-stratum imbalances}
\end{abstract}

\section{Introduction}
    In clinical trials evaluating a new treatment, randomization is used to balance both observed and unobserved confounding variables. The simple randomization, which assigns each subject to one of the treatment groups  independently with a fixed probability, is the most popular, but, in practice, may lead to an imbalance among the number of subjects in treatment groups by chance. Adaptive randomization which allocates subjects sequentially based on the previous assignments are proposed to reduce the imbalance; examples include the biased coin design \citep{efron}, the permuted block design \citep{zelen1974randomization},  and the urn design \citep{wei1978adaptive}. Further, to reduce bias in comparison between treatment groups, practitioners usually hope for a balance among some prespecified and important stratification factors or covariates, such as age and severity of illness, which can be achieved by covariate-adaptive randomization. One approach is to apply one of the above adaptive randomization approaches to each stratum of the stratification factors, which leads to the stratified version of the biased coin, urn, block designs; see, e.g., \cite{breugom2015adjuvant,stott2017thyroid}.  The minimization method, first proposed by \cite{taves} and extended by \cite{pocock}, is another widely used approach especially when the number of covariate strata is large;  see, e.g., \cite{hunt1992isis,van2010randomized}. Instead of balanced allocation within each stratum, it aims for marginal balance for the stratification factors and is thus also known as the marginal method. The popularity of the minimization method is noted in \cite{taves2} that it has been used over $500$ times during $1989$-$2008$.

For data collected under covariate-adaptive randomization, due to correlated treatment assignments, adjustments are required for inferential procedures which are developed for the simple randomization \citep{shao,xu2016validity,bugni2018inference,bugni2019inference,hu,ye,ma2020statistical,li,bugni2023inference,ye2022inference,liu2023lasso}. 
    Specifically, denote by $\bm{Z}_1,\ldots,\bm{Z}_n$
    a sequence of discrete vectors that are used in covariate-adaptive randomization. For instance, in \citet[Table 1]{alphonso2005prospective}, a study on lung resection uses the minimization method and controls for several factors including smoking status, age and sex, each of which has two levels $\{0,1\}$ representing smokers/non-smokers, age$>$median/age$\leq$median and males/females respectively; then $\bm{Z}_1 = (1,1,0)$ indicates that the first subject is a young male non-smoker. 
    These covariate vectors  are independent with a common probability mass function (pmf) $\bm{p}_0$ defined over all the strata, $\mathcal{Z}$. Further, for each $\bm{z} \in \mathcal{Z}$, let $\bm{S}_n(\bm{z})$ be the difference between the number of subjects in the treatment groups among those within stratum $\bm{z}$, after the assignments of  $n$-subjects.  Finally, denote by $\bm{\Sigma}_n$ the covariance matrix of the within-stratum imbalances $n^{-1/2}\bm{S}_n$,  where $\bm{S}_n = (\bm{S}_n(\bm{z}): \bm{z} \in \mathcal{Z})$,
    and by $\bm{\Sigma}_\infty$ the large-$n$ limit of $\bm{\Sigma}_n$ if it exists. 
    A  consistent estimator for the limiting covariance matrix $\bm{\Sigma}_\infty$ is needed to develop valid robust  tests for treatment effects, e.g., under generalized linear models with omitted variables \citep{li} or with survival data under possibly mis-specified models \citep{ye,johnson2022validity}.
    
    For most designs of the covariate-adaptive randomization,  the limit $\bm{\Sigma}_\infty$ exists and takes the form $v \times \text{diag}(\bm{p}_0)$, where  $v$ is a design-specific scalar, and 
    $\text{diag}(\bm{p}_0)$  a diagonal matrix with the diagonal vector being $\bm{p}_0$.
    Specifically, it is known that $v = 0$ for the stratified permuted block design and biased coin design \citep{efron,baldi}, $v = 1/3$ for the stratified urn design \citep{wei1988properties}, and $v = 1$ for the simple randomization.  
    However, for the minimization method, the existence of $\bm{\Sigma}_\infty$ is established only recently by \cite{hu}. Specifically, it is proved in \cite{hu} that for each  $\bm{z} \in \mathcal{Z}$, the within-stratum variance, $\text{var}(n^{-1/2}\bm{S}_n(\bm{z}))$, converges to a stratum specific limit $\sigma^2_{{z}}$. With similar techniques and the Cram\'er-Wold Device,
    one may show that the covariance matrix of $n^{-1/2}\bm{S}_n$ indeed converges to $\bm{\Sigma}_\infty$, whose computation, however, requires solving Poisson equations associated with an induced  Markov chain, and is challenging even when the dynamics of the Markov chain are given; see, e.g., \cite{douc2022solving}. 
    
    Our main contribution in this paper is in proposing a consistent estimator for $\bm{\Sigma}_\infty$ under the minimization method, without assuming the pmf $\bm{p}_0$  known. Indeed, if $\bm{p}_0$ is given, one may estimate the covariance matrix $\bm{\Sigma}_n$  by Monte Carlo simulations. Specifically, for each of $B$ independent repetitions, we generate $n$ observations from the pmf $\bm{p}_0$, make assignments for these subjects according to the minimization method, and calculate the imbalances $\bm{S}_n$; finally, we estimate $\bm{\Sigma}_n({p}_0)$ by the empirical covariance matrix $\hat{\bm{\Sigma}}_{n}^{B}(\bm{p}_0)$ of $n^{-1/2}\bm{S}_n$ from the $B$ repetitions.
    However, $\bm{p}_0$ is a population-level quantity that is usually unknown and has to be estimated in practice.  We propose to replace $\bm{p}_0$ in the above procedure with some estimator $\hat{\bm{p}}_n$. We establish that if $n^{1/2}(\hat{\bm{p}}_n - \bm{p}_0)$ converges in distribution, then this procedure produces a consistent estimator for the limiting covariance matrix, that is,
     $\hat{\bm{\Sigma}}_{n}^{B}(\hat{\bm{p}}_n) \to \bm{\Sigma}_\infty$ in probability as  both $n$ and $B$ diverge to infinity. One obvious choice for $\hat{\bm{p}}_n$ would be the empirical pmf of $\bm{Z}_1,\ldots,\bm{Z}_n$, while 
     more efficient estimators may be obtained under some structural assumptions, such as independence among stratification factors, or with access to additional data. In terms of techniques, our proofs rely on Le Cam's third lemma, the Lyapunov techniques in \cite{hu} and properties of Poisson equation solutions associated with symmetric Markov chains.

    As an application, we  consider time-to-event data collected under the minimization method and the adjustments to the robust tests for treatment effects proposed by \cite{lin}. \cite{ye} and \cite{johnson2022validity} developed an asymptotically valid adjustment that requires a consistent estimator for $\bm{\Sigma}_{\infty}$. Specifically, under the assumption that each stratum has the same prevalence, in which case the pmf $\bm{p}_0$ is \textit{known}, \cite{johnson2022validity} establishes a structural result for $\bm{\Sigma}_{\infty}$, and estimates it by a Monte Carlo approach. 
    In this work, we consider general, unknown $\bm{p}_0$, and use the proposed  consistent estimator $\hat{\bm{\Sigma}}_{n}^{B}(\hat{\bm{p}}_n)$ for adjustment. In simulations, we observe that the size  of the adjusted test from our approach   is close  to the nominal level for a moderate sample size $n$.  
    Further, again by numerical studies, we show that the robust tests \citep{lin} without adjustment may have an asymptotically inflated size under the minimization method. This is because unlike other designs, the difference between the limiting covariance matrices under the simple randomization  and the minimization method may not be positive semidefinite, which we also confirm by simulations.
    
    The paper is organized as follows. In Section \ref{sec:mdprop}, we recall the minimization method, propose under this setup an estimator for the limiting covariance matrix of the within-stratum imbalances and further show that the estimator is asymptotically consistent. In Section \ref{sec:survival}, we consider adjustments to statistical tests in survival analysis using the proposed estimator and conduct simulation studies to evaluate their finite sample performance in terms of  size and power. 
    Omitted proofs are provided in Appendix.


\section{Minimization method and consistent estimators for limiting covariance matrix} \label{sec:mdprop}
\subsection{Minimization method} \label{subsec:designnotationsandprocedure}

Consider $n$ subjects with $K$-dimensional \textit{discrete} covariate vectors $\bm{Z}_1,\ldots, \bm{Z}_n$, which are used in making treatment assignments  and assumed to be independent and identically distributed. We denote by $\bm{Z}$ a generic covariate vector,
and use superscripts to index components of a vector, for instance, $\bm{Z} = (\bm{Z}^1,\ldots,\bm{Z}^K)$. 
For each stratification factor $k \in \{1,\ldots, K\} = [K]$, denote by  $\mathcal{Z}^{k}$ the levels for the $k$-th factor $\bm{Z}^k$, and by $m_k = |\mathcal{Z}^{k}|\geq 2$ its size. 
Further, we denote by $\mathcal{Z} = \mathcal{Z}^1\times \ldots\times \mathcal{Z}^k$  all strata of $\bm{Z}$, and by $m = m_1 \times \ldots \times m_K$ the number of strata, which is assumed to be fixed and does not change with the sample size $n$.
Finally, we denote by $\bm{p}_0$ the probability mass function  of $\bm{Z}$ and assume, for simplicity, that
$\bm{p}_0(\bm{z})= \mathbb{P}(\bm{Z}=\bm{z}) > 0$ for each $\bm{z}\in \mathcal{Z}$;  otherwise, we could drop those strata from $\mathcal{Z}$ with zero probability.

We consider a clinical trial with two treatment groups. For each subject $i\in [n]$, denote by $I_i$  the indicator for its assignment: $I_i=1$ if the $i$-th subject is assigned to (active) treatment, and $I_i = 0$ if control. The minimization method \citep{pocock} allocates subjects \textit{sequentially} based on the previous assignments and marginal imbalances. Specifically,  for $i \in [n]$, assume that assignments for the first $(i-1)$ subjects have been completed, and we consider the following procedure for allocating the $i$-th subject.

\begin{itemize}[leftmargin=*]
    \item[] Step 1. For each $k \in [K]$ and  $\bm{z}^k \in \mathcal{Z}^k$, compute the marginal imbalance, $\bm{M}_{i-1}(k,\bm{z}^k)$, for the $k$-th stratification factor at level $\bm{z}^k$  after the completion of  first $(i-1)$ assignments, i.e.,
    \begin{equation}
    \label{marginalimb}
       \bm{M}_{i-1}(k, \bm{z}^k) = \sum_{j\leq i-1}\mathbb{1}(\bm{Z}_j^k=\bm{z}^k)\{I_j - (1-I_{j})\},
    \end{equation}
   with the convention that $\bm{M}_{0}(k,\bm{z}^k) = 0$, where $\mathbb{1}\left(\cdot\right)$ denotes the indicator function.
   
   \item[] Step 2. Given a  \textit{positive} weight vector $\bm{\omega} = (\bm{\omega}^1,\ldots,\bm{\omega}^K)$, compute  the potential weighted imbalance 
   measure $\text{Imb}_i^{(0)}$ (resp.~$\text{Imb}_i^{(1)}$)  if the $i$-th subject is assigned to control (resp.~treatment), i.e.,
   \begin{equation}\label{Imbalance measure}
        \text{Imb}_i^{(0)} = \sum_{k=1}^{K}\bm{\omega}^k\{\bm{M}_{i-1}(k,\bm{Z}_i^k)-1\}^2,
        \quad \text{Imb}_i^{(1)} = \sum_{k=1}^{K}\bm{\omega}^k\{\bm{M}_{i-1}(k,\bm{Z}_i^k)+1\}^2.
  \end{equation}

     \item[] Step 3. Assign the $i$-th subject to the treatment group with probability $g(\text{Imb}^{(1)}_i-\text{Imb}^{(0)}_i)$, where $g:\mathbb{R} \to (0,1)$ is a user-given function such that  $g(x)=1-g(-x)$ for $x \in \mathbb{R}$,  $g(x)\leq 0.5$ if $x\geq 0$ and $\lim\sup_{x\to\infty} g(x)<0.5$. 
\end{itemize}

\begin{remark}
A popular choice for the function $g$ is as follows: for $q < 1/2$,
\begin{equation}\label{def:gfunction}
        g_q(x) = q \mathbb{1}(x > 0) + 0.5\times \mathbb{1}(x = 0) + (1-q) \mathbb{1}(x < 0),\quad \text{ for } x \in \mathbb{R}.
    \end{equation}
For further discussions, see \citet[Remark 2.1]{hu}. We note that the requirement $g(x)=1-g(-x)$ is needed for the analysis. For the other conditions on $g(\cdot)$,  they are intuitive in that if assigning a subject to treatment $1$ induces a more (resp. less) severe potential imbalance, then the probability for treatment $1$ should be lower (resp. higher) than treatment $0$.
\end{remark}

As discussed in the Introduction section, the key quantity to develop valid tests under the minimization method is the limiting covariance matrix of the (scaled) within-stratum imbalances. Specifically, for each stratum $\bm{z} \in \mathcal{Z}$, denote by $\bm{S}_n(\bm{z})$ its imbalance under the minimization method, i.e., 
\begin{equation}\label{stratum imbalance}
          \bm{S}_n(\bm{z})= \sum_{i=1}^{n}\mathbb{1}(\bm{Z}_i=\bm{z})\{I_i-(1-I_i)\}.
\end{equation}
Further, denote by $\bm{\Sigma}_n(\bm{p}_0)$ the covariance matrix of $n^{-1/2}\bm{S}_n$, where $\bm{S}_n = (\bm{S}_n(\bm{z}), \bm{z} \in \mathcal{Z})$. In the next subsection, we show that $\bm{\Sigma}_n(\bm{p}_0)$ converges as $n \to \infty$, and propose a consistent estimator for the limit.

\textbf{Notations.} Note that once the weight vector $\bm{\omega}$ and the function $g$ are given,  the joint distribution of $\{\bm{Z}_i,I_i:i\in [n]\}$ only depends on the probability mass function (pmf) of the generic vector of stratification factors $\bm{Z}$. To emphasize this dependence, when the pmf of $\bm{Z}$ is $\bm{p}$, we denote by $\Rightarrow_{\bm{p}}$ the  convergence in distribution, and by $\bm{\Sigma}_n(\bm{p})$ the covariance matrix of $n^{-1/2} \bm{S}_n$. 

\begin{remark}
In \citet[Section 2]{hu}, a more general version of the minimization method is considered. 
Specifically, for $i \in [n]$, let 
$S_{i-1}(\bm{z}) = \sum_{j \leq i-1} \mathbb{1}(\bm{Z}_j=\bm{z}) (2I_j-1)$  be the difference between the number of subjects within stratum $\bm{z}$ in each of the two treatment groups after the assignment of the $(i-1)$-th subject, and $D_{i-1} = \sum_{\bm{z} \in \mathcal{Z}} S_{i-1}(\bm{z})$ the overall difference.
Then define for $k \in \{0,1\}$,
$$\text{Within-Imb}_{i}^{(k)} = \omega_{s}\left(\bm{S}_{i-1}(\bm{Z}_i)-(-1)^{k}\right)^2,\quad \text{and} \quad \text{Overall-Imb}_i^{(k)} = \omega_o \left(D_{i-1}-(-1)^{k}\right)^2,$$
where $\omega_s$ and $\omega_o$ are two non-negative weights. 
In \cite{hu}, they define the following potential imbalance measures: 
$\text{Overall-Imb}_i^{(k)}+\text{Imb}_i^{(k)}+ \text{Within-Imb}^{(k)}_{i}$, for $k \in \{0,1\}$,
where $\text{Imb}_i^{(k)}$ is defined in Equation \eqref{Imbalance measure}. In contrast,  we do not include the within-stratum imbalance (i.e., $\omega_s = 0$) or the overall imbalance (i.e., $\omega_o = 0$) in the calculation of the potential imbalance measures in Equation \eqref{Imbalance measure}. 

Indeed, if the within-stratum imbalance is included (i.e., $\omega_s > 0$), it is known that the limiting covariance matrix has all zero entries (Hu and Zhang, 2020, Theorem 3.1(i)). Further, our results extend readily to the case where $w_o > 0$, at the expense of additional notations. 
\end{remark}

\subsection{Consistent covariance matrix estimator}\label{subsec:covestimate}

The next theorem establishes that the within-stratum imbalance vector, $n^{-1/2}\bm{S}_n$, 
converges in distribution, and its covariance matrix admits a large-$n$ limit. Let $\bm{\Sigma}_\infty(\bm{p}_0)$ be the $m \times m$ matrix defined in Equation \eqref{def:app_sigma_infty} in Appendix \ref{app:proof_joint_local}. Its definition relies on solutions to Poisson equations \citep[Equation (6.11)]{hu}, and is presented in Appendix due to its complicated form. Finally, denote by
$\mathcal{N}_m\left(\bm{v},\bm{V}\right)$  the $m$-dimensional normal distribution with  mean vector $\bm{v}$ and covariance matrix $\bm{V}$.

\begin{theorem} \label{theorem:preliminary}
As $n \to \infty$, we have
     \begin{align*} 
     n^{-1/2} \bm{S}_n \;\Rightarrow_{\bm{p}_0}\; \mathcal{N}_m\left(\bm{0}_m,\bm{\Sigma}_\infty(\bm{p}_0)\right),\quad \text{ and } \quad  \bm{\Sigma}_n(\bm{p}_0)\to \bm{\Sigma}_\infty(\bm{p}_0),
     \end{align*}
     where $\bm{0}_m$ is the $m$-dimensional zero vector.
 \end{theorem}

 \begin{proof} 
See Subsection \ref{subsec:proof_strategy}. The proof technique is essentially due to \citet{hu}.
 \end{proof}

\begin{remark}
   In \citet[Theorem 3.2(iv)]{hu}, they prove the marginal convergence within each stratum, that is, for each $\bm{z}\in \mathcal{Z}$, $n^{-1/2}\bm{S}_n(\bm{z})$ converges in distribution to a zero mean normal distribution. Specifically, they show that $\bm{S}_n(\bm{z})$ is approximated by a martingale (see the decomposition in \citet[equation (6.15)]{hu}) and then apply central limit theorems for martingales.

   In the above theorem, we establish the joint convergence using the Cram\'er-Wold device. Specifically,  we show, via the same martingale argument as in \citet[Theorem 3.2]{hu}, that for any $m$-dimensional vector $\bm{a} = (\bm{a}_{\bm{z}}:\bm{z}\in \mathcal{Z})$, $n^{-1/2}\sum_{\bm{z} \in \mathcal{Z}} \bm{a}_{\bm{z}} \bm{S}_n(\bm{z})$ converges in distribution to a zero-mean normal distribution with variance $\bm{a}^T \bm{\Sigma}_{\infty}(\bm{p}_0)\bm{a}$. For details, see the proof of Lemma \ref{lemma:joint} in Appendix \ref{app:proof_joint_local}.
\end{remark}

If $\bm{p}_0$ is known, we can obtain the sampling distribution, and thus the covariance matrix, of $n^{-1/2} \bm{S}_n$   by Monte Carlo methods; see discussions below. Further, the second statement in the above theorem shows that this would produce a consistent estimator for $\bm{\Sigma}_{\infty}(\bm{p}_0)$. However, $\bm{p}_0$ is rarely known in practice  since it is a population-level quantity, with the population being the targeted patients of the new treatment, which may change over time and  be different across locations. 
Thus, we propose to use an estimator $\hat{\bm{p}}_n$ instead, as justified by the following theorem. We denote by  $\bm{\Delta}_m$ the $m$-dimensional probability simplex, $\{\bm{p}\in \mathbb{R}^m: \sum_{\bm{z} \in \mathcal{Z}}\bm{p}(\bm{z})=1,\; \bm{p}(\bm{z}) \geq 0 \text{ for each } \bm{z} \in \mathcal{Z}\}$.

\begin{theorem}\label{theorem:main}
Let $\hat{\bm{p}}_n \in \bm{\Delta}_m$ be a sequence of estimators such that $n^{1/2}(\hat{\bm{p}}_n - \bm{p}_0)$ converges in distribution as $n \to \infty$. Then as $n \to \infty$,
    $$
        \bm{\Sigma}_n(\hat{\bm{p}}_n)\to \bm{\Sigma}_\infty(\bm{p}_0) \quad\text{in probability},$$
    where $\bm{\Sigma}_\infty(\bm{p}_0)$ appears in Theorem  \ref{theorem:preliminary} and is defined in  Equation \eqref{def:app_sigma_infty} in Appendix \ref{app:proof_joint_local}.
\end{theorem}

\begin{proof}
    See Subsection \ref{subsec:proof_strategy}.
\end{proof}

The requirement that $n^{1/2}(\hat{\bm{p}}_n - \bm{p}_0)$ converges in distribution is not stringent. In particular, it holds if   $\hat{\bm{p}}_n$ is the empirical probability mass function (pmf) for the vectors of stratification factors used in the minimization method, that is, $ \hat{\bm{p}}_n(\bm{z}) =  n^{-1}\sum_{i=1}^{n}\mathbb{1}(\bm{Z}_i=\bm{z})$ for each $\bm{z}\in \mathcal{Z}$, which is usually available.

When the number of subjects $n$ does not far exceed the strata size $m$, the empirical pmf may not be accurate for estimating $\bm{p}_0$. In this case, one remedy, if at all possible, is to recruit additional $r$ subjects from the same population, from which $\bm{Z}_1,\ldots, \bm{Z}_n$ are drawn, and record only their stratification factors. Then we may use the empirical pmf $\hat{p}_N$ of the $N = n + r$ covariate vectors, including $\bm{Z}_1,\ldots, \bm{Z}_n$,  as an estimate for $\bm{p}_0$.
Further, in some applications, there may be reasonable structural assumptions on the vectors of stratification factors. For example, if the stratification factors are assumed to be independent, it suffices to estimate marginal pmfs, that is, 
$\hat{\bm{p}}_n(\bm{z}) = \prod_{k \in [K]} \{n^{-1} \sum_{i=1}^{n} \mathbb{1}(\bm{Z}_i^{k} = \bm{z}^k)\}$ for each $\bm{z} \in \mathcal{Z}$.

\textbf{Monte Carlo Estimation.} 
Given an estimator $\hat{\bm{p}}_n$ for $\bm{p}_0$, by Monte Carlo methods, one may get an estimator for the covariance matrix, $\bm{\Sigma}_n(\hat{\bm{p}}_n)$, of $n^{-1/2}\bm{S}_{n}$, when the vector of stratification factors has a pmf given by $\hat{\bm{p}}_n$.

Specifically, let $B$ be a large number, indicating the total number of repetitions. For the $b$-th repetition with $b \in [B]$, we generate a sequence of vectors of stratification factors $\{\bm{Z}_i^{(b)}: i \in [n]\}$ from the pmf $\hat{\bm{p}}_n$, follow the minimization method to make assignments, and compute the within-stratum imbalances  $\bm{S}_n^{(b)}$. Finally, we report the empirical covariance matrix of $\{n^{-1/2} \bm{S}_n^{b}: b \in [B]\}$ as our estimator, that is,
\begin{align} 
\label{def:boot_est}
    \hat{\bm{\Sigma}}_n^{B}(\hat{\bm{p}}_n) =
    (B-1)^{-1} \sum_{b=1}^{B} \left\{n^{-1/2}\bm{S}_n^{(b)}- \hat{\bm{\mu}}^{B}(\hat{\bm{p}}_n)\right\} \left\{n^{-1/2}\bm{S}_n^{(b)}  - \hat{\bm{\mu}}^{B}(\hat{\bm{p}}_n)\right\}^\text{T},
\end{align}
where $\hat{\bm{\mu}}^{B}(\hat{\bm{p}}_n) = B^{-1} \sum_{b=1}^{B} (n^{-1/2}\bm{S}_n^{(b)})$ is the empirical mean vector. 

The Monte Carlo approach could be viewed as a parametric bootstrap, and since $\bm{Z}$ is discrete, if $\hat{\bm{p}}_n$ is the empirical pmf, it is also a non-parametric bootstrap  \citep{shao2012jackknife}.
The next theorem establishes the consistency of the Monte Carlo estimator.

\begin{theorem}\label{thm:bootstrap_consistency}
Consider the setup in Theorem \ref{theorem:main}. As $\min\{n,B\} \to \infty$,
\begin{align*}
   \hat{\bm{\Sigma}}_n^{B}(\hat{\bm{p}}_n) \;\; \to \;\; \bm{\Sigma}_\infty(\bm{p}_0) \quad\text{in probability}.
\end{align*}
\end{theorem}

\begin{proof}
Note the following decomposition
$$
\hat{\bm{\Sigma}}_n^{B}(\hat{\bm{p}}_n) - \bm{\Sigma}_{\infty}(\bm{p}_0) = 
\{\hat{\bm{\Sigma}}_n^{B}(\hat{\bm{p}}_n) - {\bm{\Sigma}}_n(\hat{\bm{p}}_n)\}
+\{{\bm{\Sigma}}_n(\hat{\bm{p}}_n) - \bm{\Sigma}_{\infty}(\bm{p}_0)\}.
$$
By Theorem \ref{theorem:main}, the  term in the second bracket converges to zero in probability as $n \to \infty$. It remains to bound the (conditional) moments of the Monte Carlo error; see Appendix \ref{app:proof_bootstrap_consistency}.
\end{proof}

\begin{remark}
In the bootstrap literature, it is customary not to consider the Monte Carlo error, which usually can be  made arbitrarily small by increasing $B$. The above Theorem requires a proof because the bootstrap statistics $\{n^{-1/2}\bm{S}_n^{(b)}: b \in [B]\}$ are not bounded.    
\end{remark}

\begin{remark}
As mentioned before, a consistent estimator for $\bm{\Sigma}_{\infty}(\bm{p}_0)$ is required to develop asymptotically valid statistical tests for various models; see, e.g., \cite{li,ye,johnson2022validity}. 
In simulation studies in Section \ref{sec:survival}, we consider  adjusted tests using the proposed estimator, $\hat{\bm{\Sigma}}_n^{B}(\hat{\bm{p}}_n)$, for survival analysis.
\end{remark}

\subsection{Proof strategy}\label{subsec:proof_strategy}
Since the proof is lengthy and involved, 
in this subsection, we discuss the proof strategy and highlight the main novelty of the proof, while details can be found in Appendix.  
Recall that $\bm{\Sigma}_n(\bm{p})$ is the covariance matrix of $n^{-1/2} \bm{S}_n$ when $\bm{Z}$ has the probability mass function (pmf) $\bm{p} \in \bm{\Delta}_m$. Note that for a fixed $n$, $\bm{\Sigma}_n(\bm{p})$ may be viewed as a complicated function of $\bm{p}$.
For each pmf $\bm{p} \in \bm{\Delta}_m$, define the log-likelihood ratio of $\bm{p}$ against $\bm{p}_0$ for the sequence $\bm{Z}_1,\ldots,\bm{Z}_n$ as follows:
\begin{equation}
        \ell_n(\bm{p};\bm{p}_0) = \sum_{i=1}^{n}\log\left\{
        \bm{p}(\bm{Z}_i)/\bm{p}_0(\bm{Z}_i)\right\}.
        \label{ell}
\end{equation}

The proof involves multiple steps, which are summarized in Table \ref{tab:step} and discussed in more detail below. For the purpose of discussion, fix an arbitrary sequence of deterministic pmfs $\bm{p}_n \in \bm{\Delta}_m$ such that $n^{1/2}(\bm{p}_n-\bm{p}_0)$ is convergent as $n \to \infty$. The sequence $\{\bm{p}_n\}$ is known as local alternatives in the literature \citep{van}.

First, we show in Lemma \ref{lemma:LAN} that the likelihood ratios $\ell_n(\bm{p}_n; \bm{p}_0)$ enjoys the local asymptotic normality (LAN) at $\bm{p}_0$ \citep[Theorem 7.2]{van}. 

Second, based on the LAN property, we apply the same martingale argument as in the proof of  \citet[Theorem 3.2(iv)]{hu} and establish that the statistics $n^{-1/2} \bm{S}_n$ and the log-likelihood ratios  $\ell_n(\bm{p}_n;\bm{p}_0)$ converge in distribution \textit{jointly}.

Third, the joint convergence, together with the LAN property, allows us to derive the limiting distribution of the statistics $n^{-1/2} \bm{S}_n$ under the local alternatives  $\{\bm{p}_n\}$. Specifically, in Lemma \ref{lemma:perturbed normality}, by  Le Cam's third lemma \cite[Theorem 6.6]{van}, we prove that the distribution of $n^{-1/2} \bm{S}_n$, when $\bm{Z}$ has the pmf $\bm{p}_n$,
 converges in distribution to $\mathcal{N}_m\left(\bm{0},\bm{\Sigma}_\infty(\bm{p}_0)\right)$. Further, by \cite[Theorem 2.20]{van}, the weak convergence implies that the first two moments of  $n^{-1/2} \bm{S}_n$, when $\bm{Z}$ has the pmf $\bm{p}_n$, converges to those of the limiting distribution. 
 Since the covariance matrix of  $n^{-1/2} \bm{S}_n$ under the local alternative $\bm{p}_n$ is $\bm{\Sigma}_n(\bm{p}_n)$, we have that 
 $\bm{\Sigma}_n(\bm{p}_n)$ converges to $\bm{\Sigma}_\infty(\bm{p}_0)$. This step is the main novelty of our proof in that \textit{we establish  the property of  
  $\bm{\Sigma}_n(\bm{p}_n)$ by working with the distribution of $n^{-1/2} \bm{S}_n$ under the local alternative $\bm{p}_n$}.

Finally, given the established properties of the functions $\{\bm{\Sigma}_n(\cdot)\}$, we apply the extended continuous mapping theorem \cite[Theorem 7.24]{kosorok}  to show the consistency of the proposed procedure as stated in Theorem \ref{theorem:main}.

Next, we provide precise statements for the first three steps and prove Theorem \ref{theorem:preliminary} and  \ref{theorem:main} using  the result from the third step. We start with the asymptotic expansion for a sequence of log-likelihood ratios under local alternatives at $\bm{p}_0$.

\begin{table}[t!]
    \caption{For each step, the first column indicates the key techniques, while the second lists the precise result.}
    \centering
    \begin{tabular}{l|l}
    Step $1$: local asymptotic normality & Lemma \ref{lemma:LAN}:   asymptotic expansion of $\ell_n(\bm{p}_n;\bm{p}_0)$
    \\
      Step $2$: Cramer-Wold device, martingale CLT & Lemma \ref{lemma:joint}: joint convergence of $(n^{-1/2}\bm{S}_n,\ell_n(\bm{p}_n;\bm{p}_0))^T$.\\ Step $3$: Le Cam's third lemma,  moments convergence & Lemma \ref{lemma:perturbed normality}: convergence of $n^{-1/2}\bm{S}_n$ when $\bm{Z}$ is from $\bm{p}_n$.
      \\
      Step $4$: extended continuous mapping & Theorem \ref{theorem:main}: consistent covariances estimation.
    \end{tabular}
    \label{tab:step}
\end{table}


\begin{mylem} \label{lemma:LAN}
  Let $\bm{p}_n \in \bm{\Delta}_m$ be a sequence of deterministic probability mass functions such that $n^{1/2}(\bm{p}_n - \bm{p_0})$ converges to some vector $\bm{\delta} \in \mathbb{R}^{m}$ as $n \to \infty$. Then we have
    $$
        \ell_n(\bm{p}_n;\bm{p}_0) =  n^{-1/2} \sum_{i=1}^{n} \frac{\bm{\delta}(\bm{Z}_i)}{\bm{p}_0(\bm{Z}_i)} -\frac{1}{2} \sum_{\bm{z} \in \mathcal{Z}}
        \frac{\bm{\delta}^2(\bm{z})}{\bm{p}_0(\bm{z})}
        +o_{\bm{p}_0}(1),
   $$
   where $o_{\bm{p}_0}(1)$ is a sequence of random variables that converges to $0$ in probability.
\end{mylem}

\begin{proof}
The result is essentially due to \citet[Lemma 7.6]{van}.    See Appendix \ref{app:lam} for details.
\end{proof}

Next, we show the joint convergence in distribution of the statistics
$n^{-1/2}\bm{S}_n$ and the local likelihood ratios.

\begin{mylem} \label{lemma:joint}
For any sequence of deterministic probability mass functions $\bm{p}_n \in \bm{\Delta}_m$, if $n^{1/2}(\bm{p}_n - \bm{p}_0)$ converges to some vector $\bm{\delta} \in \mathbb{R}^{m}$ as $n \to \infty$, then we have
   $$
        \begin{Bmatrix} n^{-1/2}\bm{S}_n \\ \ell_n(\bm{p}_n;\bm{p}_0)\end{Bmatrix}
       \;\; \Rightarrow_{\bm{p}_0}\;\;     
        \mathcal{N}_{m+1}\left(\begin{pmatrix}\bm{0}_m\\
        -\tau^2/2
        \end{pmatrix}
        ,\begin{pmatrix}\bm{\Sigma}_\infty(\bm{p}_0) &\bm{0}_m \\ \bm{0}_m^\text{T} &
        \tau^2         \end{pmatrix}
        \right),
    $$
    where $\tau^2 = \sum_{\bm{z} \in \mathcal{Z}} \bm{\delta}^2(\bm{z})/\bm{p}_0(\bm{z})$ and
    $\bm{\Sigma}_\infty(\bm{p}_0)$ is defined in  Equation \eqref{def:app_sigma_infty} in Appendix \ref{app:proof_joint_local}.
\end{mylem}

\begin{proof}
    See Appendix \ref{app:proof_joint_local}. 
\end{proof}

\begin{remark}
Due to the Cram\'er-Wold device,  the above lemma is equivalent to the following: for any  vector 
${\bm{a}} = \{\bm{a}_{\bm{z}}: \bm{z} \in \mathcal{Z}\}$ and scalar $b_{\ell}$,
\begin{align*}
   \sum_{\bm{z} \in \mathcal{Z}} \bm{a}_{\bm{z}} ( n^{-1/2} \bm{S}_n(\bm{z})) + b_{\ell} \left(\ell_n(\bm{p}_n;\bm{p}_0) + \tau^2/2\right)\;\; \Rightarrow_{\bm{p}_0}\;\; 
    \mathcal{N}_1\left(0, {\bm{a}}^\text{T} \bm{\Sigma}_\infty(\bm{p}_0) {\bm{a}} + b_{\ell}^2 \tau^2 \right).
\end{align*}
As in the remark following Theorem \ref{theorem:preliminary}, it is shown in the proof of  \citet[Theorem 3.2(iv)]{hu} that for each $\bm{z} \in \mathcal{Z}$, $\bm{S}_n(\bm{z})$ is approximated by a martingale. Further, by Lemma \ref{lemma:LAN}, $\ell_n(\bm{p}_n;\bm{p}_0) + \tau^2/2$ is asymptotically equivalent to a sum of i.i.d.~zero mean random variables, which is also a  martingale. Then we apply  martingale central limit theorems to establish the above convergence, which relies on techniques developed in \cite{hu}, and a structural  property of Poisson equation solutions associated with symmetric Markov chains (see Lemma \ref{lemma:odd function} in Appendix).
\end{remark}


In Appendix \ref{app:lam}, we show in Lemma \ref{lemma:changeofmeasure} that $\ell_n(\bm{p};\bm{p}_0)$ is, in fact, also the log-likelihood ratio of $\bm{p}$ against $\bm{p}_0$ for the joint sequence $\{\bm{Z}_i, I_i: i \in [n]\}$. Then by Le Cam's third lemma  \cite[Theorem 6.6]{van}, using the LAN property in Lemma \ref{lemma:LAN}, we establish the limiting distribution of $n^{-1/2} \bm{S}_n$ under local alternatives $\{\bm{p}_n\}$.

\begin{mylem} \label{lemma:perturbed normality}
For any sequence of deterministic probability mass functions $\bm{p}_n \in \bm{\Delta}_m$, if $n^{1/2}(\bm{p}_n - \bm{p}_0)$ converges to some vector as $n \to \infty$, then we have
    $$
        n^{-1/2}\bm{S}_n \;\;\Rightarrow_{\bm{p}_n}\;\; \mathcal{N}_m\left(\bm{0},\bm{\Sigma}_\infty(\bm{p}_0)\right), \quad \text{ and } \quad
        \bm{\Sigma}_n(\bm{p}_n) \to \bm{\Sigma}_{\infty}(\bm{p}_0),
    $$
    where $\bm{\Sigma}_\infty(\bm{p}_0)$ is defined in  Equation \eqref{def:app_sigma_infty} in Appendix \ref{app:proof_joint_local}, and $\Rightarrow_{\bm{p}_n}$ means that the distribution of $n^{-1/2}\bm{S}_n$, when the pmf of $\bm{Z}$ is given by $\bm{p}_n$, converges weakly as $n \to \infty$.
\end{mylem}

\begin{proof}
 See Appendix \ref{app:proof_general_local}.
\end{proof}

\begin{remark}
The second part in  Lemma \ref{lemma:perturbed normality} states the following property regarding the sequence of functions $\{\bm{\Sigma}_n(\cdot)\}$ locally at $\bm{p}_0$: for any sequence of pmfs $\{\bm{p}_n\}$, as long as $n^{1/2}(\bm{p}_n - \bm{p}_0)$ is convergent, the function $\bm{\Sigma}_n(\cdot)$ evaluated at $\bm{p}_n$ converges to the same limit $\bm{\Sigma}_\infty(\bm{p}_0)$. 
\end{remark}

Given Lemma \ref{lemma:perturbed normality}, the proofs of Theorem  \ref{theorem:preliminary} and  \ref{theorem:main} are immediate.

\begin{proof}[Proof of Theorem \ref{theorem:preliminary}]
     It follows from Lemma \ref{lemma:perturbed normality} by letting $\bm{p}_n = \bm{p}_0$ for all $n \geq 1$.
\end{proof}   

Finally, we prove the main theorem.
\begin{proof}[Proof of Theorem \ref{theorem:main}]
For each $n \geq 1$, define $\bm{D}_n = \{n^{1/2}(\bm{p}-\bm{p}_0)\in \mathbb{R}^{m} : \bm{p} \in \bm{\Delta}_m\}$, and for each $\bm{\delta} \in \bm{D}_n$, define $f_n(\bm{\delta}) = \bm{\Sigma}_n(\bm{p}_0+n^{-1/2}\bm{\delta})$. Further, for each $\bm{\delta} \in \mathbb{R}^{m}$, define $f_{\infty}(\bm{\delta}) = \bm{\Sigma}_{\infty}(\bm{p}_0)$, that is, $f_{\infty}(\cdot)$ is a constant function.

Fix any sequence of $\delta_n \in  \bm{D}_n$ such that $\bm{\delta}_n$ converges to some $\bm{\delta} \in \mathbb{R}^{m}$. Define $\bm{p}_n = \bm{p}_0 + n^{-1/2} \delta_n \in \mathbb{R}^{m}$. Clearly, $n^{1/2}(\bm{p}_n - \bm{p}_0) = \delta_n$ converges to $\delta$ as $n \to \infty$. Then by Lemma \ref{lemma:perturbed normality} (and also see its subsequent remark),
we have 
$$f_n(\bm{\delta}_n) = \bm{\Sigma}_n(\bm{p}_n) \;\; \to \;\;
\bm{\Sigma}_{\infty}(\bm{p}_0) = 
f_{\infty}(\bm{\delta}), \text{ as } n \to \infty.
$$ 
Thus the sequence of functions $\{f_n(\cdot)\}$ and $f_{\infty}(\cdot)$ meet the condition required by the extended continuous mapping theorem \cite[Theorem 7.24]{kosorok}. 

Finally, since by assumption $n^{1/2}(\hat{\bm{p}}_n - \bm{p}_0)$ converges to some random vector $\bm{\Theta}$,  \cite[Theorem 7.24]{kosorok} implies that 
$f_n(n^{1/2}(\hat{\bm{p}}_n - \bm{p}_0))\; \Rightarrow_{\bm{p}_0}\; f(\bm{\Theta}) = \bm{\Sigma}_{\infty}(\bm{p}_0)$, where the equality is due to the definition of $f_{\infty}(\cdot)$. The proof is complete since by definition $f_n(n^{1/2}(\hat{\bm{p}}_n - \bm{p}_0)) = \bm{\Sigma}_n(\hat{\bm{p}}_n)$.
\end{proof}

\begin{remark}
    As is evident from the above proof, the key step is the ``continuity" of functions $\{\bm{\Sigma}_n(\cdot)\}$ at $\bm{p}_0$, established in Lemma \ref{lemma:perturbed normality}. We note again that $\bm{\Sigma}_n(\bm{p})$ is a complicated function of $\bm{p}$. The main insight is that 
$\bm{\Sigma}_n(\bm{p}_n)$ is the covariance matrix of  $n^{-1/2} \bm{S}_n$  when $\bm{Z}$ has the pmf $\bm{p}_n$, and that we establish this notion of continuity by studying the limiting distribution of $n^{-1/2} \bm{S}_n$ under the local alternative $\bm{p}_n$.
\end{remark}


\subsection{Numerical study: estimation error of  covariance matrices}\label{subsec:simcov}

In this simulation study, we compare the proposed estimator $\hat{\bm{\Sigma}}_n^{B}(\hat{\bm{p}}_n)$ in \eqref{def:boot_est} with the covariance matrix $\bm{\Sigma}_n(\bm{p}_0)$ of the scaled within-stratum imbalances, $n^{-1/2} \bm{S}_n$, under minimization method. Specifically, for the minimization method, we set the weight vector $\bm{\omega}=(1/K,\ldots,1/K)$, and  the function $g$ to be $g_q(\cdot)$ in Equation \eqref{def:gfunction} with $q = 0.3$.

We consider the sample size $n \in \{200,500,1000\}$  and the following two setups for the stratification factors.

\textbf{Setup 1.} two stratification factors each with two levels, i.e., $m = 2 \times 2$, and  $\bm{p}_0$ is one of the following: 
$\bm{p}_0^{(1)}=(1/4,1/4,1/4,1/4);\bm{p}_0^{(2)}=(1/3,1/6,1/6,1/3);\bm{p}_0^{(3)}=(2/3,1/9,1/9,1/9)$. 

\textbf{Setup 2.} three stratification factors with levels $2,2$ and $\kappa \in \{3,5\}$ respectively, i.e., $m=2\times2\times \kappa$, and    $\bm{p}_0 = (1/m,\ldots,1/m)$. 

For the estimator $\hat{\bm{\Sigma}}_n^{B}(\hat{\bm{p}}_n)$, we let $B = 1,000$ and consider three choices for $\hat{\bm{p}}_n$ as discussed in Subsection \ref{subsec:covestimate}. The first is the empirical probability mass function (pmf) of $\{\bm{Z}_1,\ldots,\bm{Z}_n\}$, denoted by ``Emp" in the reported tables; the second assumes  independence among the stratification factors, denoted by ``Ind";  the third is the empirical pmf of a larger sample with size $N = 4n$, denoted by ``$N=4n$". For $\bm{\Sigma}_n(\bm{p}_0)$, it is estimated by $\hat{\bm{\Sigma}}_n^{B'}(\bm{p}_0)$ with $B'=5\times 10^5$ and assuming that $\bm{p}_0$ is known. Since $B'$ is large, $\hat{\bm{\Sigma}}_n^{B'}(\bm{p}_0)$ is viewed as a surrogate for $\bm{\Sigma}_n(\bm{p}_0)$, and thus we evaluate the loss of accuracy due to the estimation error of $\hat{\bm{p}}_n$  and a small number of Monte Carlo repetitions.

In Table \ref{tab1} and \ref{tab2}, we report the estimated difference $\|\hat{\bm{\Sigma}}^B_n(\hat{\bm{p}}_n)-\bm{\Sigma}_n(\bm{p}_0)\|_{\max}$ (``Sup"),  and relative difference  $\|\hat{\bm{\Sigma}}^B_n(\hat{\bm{p}}_n)-\bm{\Sigma}_n(\bm{p}_0)\|_{\max}/\|\bm{\Sigma}_n({\bm{p}}_0)\|_{\max}$ (``Rel Sup"), each based on the average over $10^4$ repetitions, where for a matrix $\bm{A}$, $\|\bm{A}\|_{\max}$ is the maximum of the absolute values of its entries.
From Table \ref{tab1}, when the number of strata $m$ is small, the proposed estimator with the empirical pmf $\hat{\bm{p}}_n$ approximates well the target $\bm{\Sigma}_n(\bm{p}_0)$, even in the case of positive correlation between the two stratification factors (i.e., $\bm{p}_0^{(2)}$) or the existence of a dominant stratum (i.e., $\bm{p}_0^{(3)}$). From Table \ref{tab2},  the  estimation error increases with $m$  and is relatively large  
when the empirical pmf is used in the smaller sample size cases, but we may improve the performance of  $\hat{\bm{\Sigma}}_n^{B}(\hat{\bm{p}}_n)$ with a more accurate $\hat{\bm{p}}_n$, which may be obtained under additional assumptions or data.

\begin{table}[bt]
\centering
\caption{Setup 1 with $m=2\times 2$ and with $\hat{\bm{p}}_n$ being the empirical pmf.}
\label{tab1}
\begin{tabular}{ccccccc}
\headrow
                &&$\bm{p}_0^{(1)}$  &  & $\bm{p}_0^{(2)}$ & & $\bm{p}_0^{(3)}$
                \\ 
\hiderowcolors
\multirow{2}{4em}{$n=200$} &\text{Sup}
& 7.62E-4  && 9.07E-4  && 8.85E-4   \\  
&\text{Rel Sup}
& 4.13E-2  && 5.38E-2  && 7.09E-2  \\ 
\multirow{2}{4em}{$n=500$} &\text{Sup}
& 6.39E-4  && 6.74E-4  && 5.60E-4   \\  
&\text{Rel Sup}
& 3.74E-2  && 4.40E-2  && 5.39E-2  \\ 
\multirow{2}{4em}{$n=1000$} &\text{Sup}
& 6.13E-4  && 5.87E-4  && 4.68E-4   \\  
&\text{Rel Sup}
& 3.67E-2  && 3.99E-2  && 4.81E-2  \\ 
\end{tabular}
\caption{Setup 2 with $m = 2\times2\times \kappa$, $\kappa = 3,5$, and $\bm{p}_0=(1/m,\ldots,1/m)^\text{T}$. The additional data case is considered with $N=4n$.
}
\label{tab2}
\begin{tabular}{ccccccccccc}
\headrow
                && \multicolumn{3}{c}{$m=2\times 2\times 3$}  &  &\multicolumn{3}{c}{$m=2\times 2\times 5$}  \\ 
                    \cline{3-5} \cline{7-9} 
                && Emp  & Ind & $N=4n$ &
                & Emp  & Ind & $N=4n$ \\ \hiderowcolors
\multirow{2}{4em}{$n=200$} &\text{Sup}
& 1.56E-3 & 1.19E-3  & 9.01E-4  && 1.54E-3 & 1.07E-3  & 8.22E-4 \\  
&\text{Rel Sup}
& 1.19E-1  & 9.08E-2  & 6.84E-2  && 1.73E-1  & 1.21E-1  & 9.28E-2  \\  
\multirow{2}{4em}{$n=500$} &\text{Sup}
& 9.96E-4  & 8.07E-4 & 6.36E-4  && 9.41E-4 & 6.82E-4   & 5.44E-4  \\  
&\text{Rel Sup}
& 7.84E-2 & 6.35E-2 & 5.03E-2  && 1.11E-1  & 8.05E-2  & 6.39E-2\\  
\multirow{2}{4em}{$n=1000$} &\text{Sup}
& 7.57E-4  & 6.51E-4   & 5.39E-4  && 6.79E-4 & 5.26E-4  & 4.36E-4 \\  
&\text{Rel Sup}
& 6.04E-2 & 5.20E-2  & 4.31E-2  && 8.14E-2  & 6.30E-2  & 5.21E-2 \\  
\hline
\end{tabular}
\end{table}

\subsection{Numerical study: variance inflation along some directions}\label{subsec:dominate}

In this subsection, our goal is to demonstrate via examples that along some directions, there is an increased  variance of within-stratum imbalances under the minimization method compared to the simple randomization, which, as we shall see in Subsection \ref{subsec:inflation}, may lead to an inflated Type-I error of statistical tests developed under the simple randomization.
 
We consider two stratification factors each with two levels, i.e., $m = 2\times 2$. For any positive $\bm{p}_0$ and general function $g(\cdot)$ in Step 3 of the minimization method, essentially by arguments in \cite{hu},  we show in Appendix \ref{app:matrix} that
\begin{equation}\label{def:specialcov}
    \bm{\Sigma}_\infty(\bm{p}_0;g) = \nu(\bm{p}_0; g) \times \begin{pmatrix}
        \bm{v}^*, -\bm{v}^*, -\bm{v}^*, \bm{v}^*
    \end{pmatrix},
\end{equation}
where $\nu(\bm{p}_0; g)$ is a scalar, $ \bm{v}^* = (1,-1,-1,1)^\text{T}$, and we make the dependence on $g(\cdot)$ explicit. Further, under the simple randomization, the limiting covariance matrix of $n^{-1/2}\bm{S}_n$ is $\text{diag}(\bm{p}_0)$,   the $m\times m$ diagonal matrix with the diagonal vector being $\bm{p}_0$.

In the simulation study, we consider $\bm{p}_0 = (1/4,1/4,1/4,1/4)$, the weight vector $\bm{\omega}=(1/2,1/2)$  and the function $g_q(\cdot)$ in \eqref{def:gfunction} with various $q \in (0.1,0.45)$; since $\bm{p}_0$ is fixed, we write $\nu(q)$ for $\nu(\bm{p}_0;g_q)$. In this case, the smallest eigenvalue for the difference, $\text{diag}(\bm{p}_0) - \bm{\Sigma}_\infty(\bm{p}_0;g_q)$, between the two limiting covariance matrices  is $1/4-4\nu(q)$, which corresponds to the eigenvector $\bm{v}^*/2$. For each $q$, we estimate $\nu(q)$ by the first entry of $\hat{\bm{\Sigma}}^{B}_n(\bm{p}_0)$ in \eqref{def:boot_est} with $B = 4\times 10^5$   for $n \in \{10^4,5\times 10^4\}$, assuming that $\bm{p}_0$ is known. 
Further, based on the $B$ Monte Carlo samples, we construct a $95\%$ confidence interval(CI) for $\nu(q)$.  The results are reported in Figure \ref{fig:dominate}, where the black lines and shaded bands indicate the estimates and $95\%$CIs of  $1/4-4\nu(q)$ respectively.

From Figure \ref{fig:dominate}, we may conclude that for
$\bm{p}_0 = (1/4,1/4,1/4,1/4)$, and  $g_q(\cdot)$ with $0.25 \leq q \leq 0.4$, the smallest eigenvalue of $\text{diag}(\bm{p}_0) - \bm{\Sigma}_\infty(\bm{p}_0;g_q)$ is negative, which implies that along some directions, such as $\bm{v}^*/2$, the variance is larger under the minimization method than under the simple randomization. This  phenomenon does \textit{not} occur in most covariate-adaptive designs, including the stratified version of the urn, permuted block, and biased coin design. 

Further, when $q \leq 0.2$, the limiting covariance matrix  $\bm{\Sigma}_\infty(\bm{p}_0;g_q) $ under the minimization method is dominated by that under the simple randomization, i.e., $\text{diag}(\bm{p}_0)$. To the best of our knowledge, providing sufficient conditions for this dominance is still an open question.

\begin{remark}    
Our simulation setup is also considered in \cite{johnson2022validity}, which establishes the structural result \eqref{def:specialcov} under the assumption that $\bm{p}_0 =(1/4,1/4,1/4,1/4)$. As mentioned above, by arguments in \cite{hu}, \eqref{def:specialcov} holds for general $\bm{p}_0$. We note that our contribution is in proposing a consistent estimator for $\bm{\Sigma}_{\infty}(\bm{p}_0)$ when $\bm{p}_0$ needs to be estimated. Further, we complement the results of \cite{johnson2022validity} in that we provide examples (i.e., large $q$), under which $\bm{\Sigma}_{\infty}(\bm{p}_0;g_q)$ is not dominated by  the limiting covariance matrix  under the simple randomization.
\end{remark}

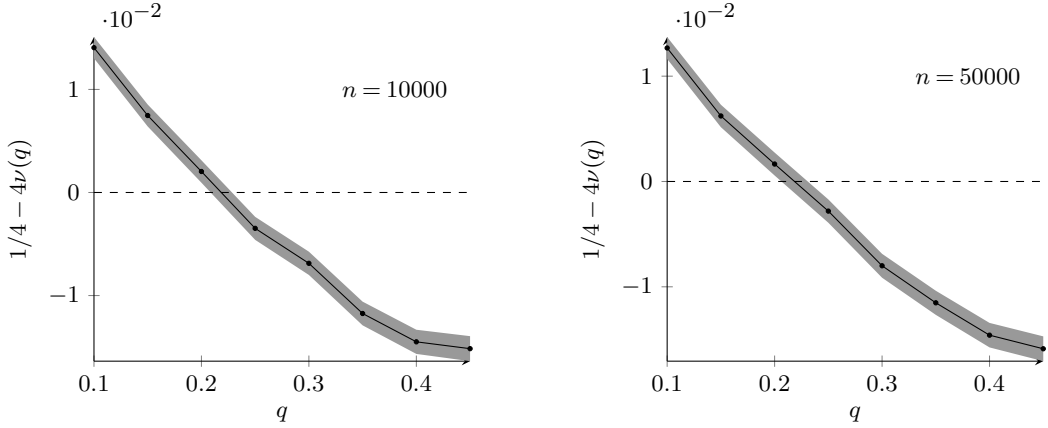
\begin{figure}[bt]
\caption{Estimates and $95\%$ confidence intervals  of $1/4-4\nu(q)$ under the setup where $m=2\times 2$ with equal prevalence $\bm{p}_0 = (1/4,\ldots,1/4)$.}
    \begin{tikzpicture}
        \begin{axis}[
            scale only axis,
            axis lines = left,
            xlabel = \(q\),
            ylabel = {\(1/4-4\nu(q)\)}
        ]
        \addplot [name path=lower, fill=none, draw=none] table [
        x=X, y expr=\thisrow{Y} - \thisrow{CI95}]{data.txt};
        \addplot [name path=upper, fill=none, draw=none] table [
        x=X, y expr=\thisrow{Y} + \thisrow{CI95}]{data.txt};
        \addplot[black!40] fill between[of=lower and upper];
        \addplot [mark=*,mark size = 0.8pt] table [x=X, y=Y]{data.txt};
        \addplot [domain=0.1:0.45,dashed]{0};
        \draw node (bar1) at (0.38,0.01){$n=10000$};
    \end{axis}

        \begin{axis}[ 
        scale only axis,
            axis lines = left,
            xlabel = \(q\),
            xshift = 3in,
            ylabel = {\(1/4-4\nu(q)\)},
        ]
        \addplot [name path=lower, fill=none, draw=none] table [
        x=X, y expr=\thisrow{Y} - \thisrow{CI95}]{data50000.txt};
        \addplot [name path=upper, fill=none, draw=none] table [
        x=X, y expr=\thisrow{Y} + \thisrow{CI95}]{data50000.txt};
        \addplot[black!40] fill between[of=lower and upper];
        \addplot [mark=*,mark size = 0.8pt] table [x=X, y=Y]{data50000.txt};
        \addplot [domain=0.1:0.45,dashed]{0};
        \draw node (bar1) at (0.38,0.01){$n=50000$};
    \end{axis}
\end{tikzpicture}
\label{fig:dominate}
\end{figure}

\section{Adjustment to robust tests for survival analysis} \label{sec:survival}

We first recall the robust tests \citep{lin} for treatment effects with survival outcomes, and its adjustment under the minimization method \citep{johnson2022validity,ye}, which requires a consistent estimator for the limiting covariance matrix $\bm{\Sigma}_{\infty}(p_0)$.

\subsection{Adjusted robust tests for treatment effect}\label{subsec:adjtest}

We start with the setup for testing treatment effects in survival analysis. For each subject, we observe 
$\bm{Z},I, X = \min(T_I,C_I),\delta= \mathbb{1}(T_I \leq C_I),\bm{W}$, 
where $\bm{Z}$ is the stratification factor used in randomization, $I$ is the binary indicator for treatment group, $\{T_0, T_1\}$ (resp.~$\{C_0,C_1\}$) are the potential survival (resp.~censoring) time for $I=0$ and $I=1$, and $\bm{W}$ covariates used for statistical inference. Further, let $\bm{V}$ be the collection of all covariates, including unobserved ones; in particular, $\bm{Z}$ and $\bm{W}$ are deterministic functions of $\bm{V}$. 

Denote by $h(\cdot \mid\bm{v},1)$ (resp. $h(\cdot \mid\bm{v},0)$) the true hazard function of $T_1$ (resp. $T_0$) when $\bm{V} = \bm{v}$, and the goal is to test  the null $H_0:h(t\mid \bm{v},1)=h(t\mid \bm{v},0)$ for all  $t \in \mathbb{R}$ and $\bm{v}$, based on $n$-observations $\{(\bm{Z}_i, I_i, X_i, \delta_{i}, \bm{W}_i): i \in [n]\}$, where $\{\bm{Z}_i, \bm{W}_i, T_{i,1}, T_{i,0}, C_{i,1},C_{i,0}: i \in [n]\}$ are independent with the same distribution as $\{\bm{Z},\bm{W}, T_0, T_1, C_0, C_1\}$, and $\{I_i: i \in [n]\}$ depend on the adopted  covariate-adaptive design.

Under the simple randomization, \cite{lin} proposes a test for $H_0$ that is robust against model mis-specification. In particular,
consider the following \textit{working} Cox proportional hazards model,
$$
    \tilde{h}(t\mid\bm{V},I) = h_0(t)\exp(\theta I+\bm{\beta^\text{T} W}), 
$$
where $h_0$ denotes a baseline hazard, and test $H_0': \theta = 0$. Under $H_0'$, denote by $\hat{\bm{\beta}}_n$ the maximum partial likelihood estimator, and $U_n(\hat{\bm{\beta}}_n)$  the usual score test statistic, where
$$
    U_n(\bm{\beta}) = n^{-1/2}\sum_{i=1}^{n}\int_{0}^{\infty} \Xi_i\left(\bm{\beta},X_i\right) dN_i(t),\;  \text{ with } \Xi_i\left(\bm{\beta},t\right)=I_i - \frac{\sum_{\ell=1}^{n}I_\ell Y_\ell(t)\exp(\bm{\beta}^\text{T}\bm{W}_\ell)}{\sum_{\ell=1}^{n}Y_\ell(t)\exp(\bm{\beta}^\text{T}\bm{W}_\ell)},
$$
and $Y_i(t) = \mathbb{1}(X_i \geq t)$, $N_{i}(t) = \delta_{i}\mathbb{1}(X_{i}\leq t)$.  Given a nominal level $\alpha$, \cite{lin} proposes to reject the null $H_0$ when $|T_{\text{R}}| \geq \Phi^{-1}(\alpha/2)$, where $\Phi$ is the cumulative distribution function of $\mathcal{N}_1(0,1)$, and the statistic  $T_{\text{R}} =U_n(\hat{\bm{\beta}}_n)/\{\hat{B}_R(\hat{\bm{\beta}}_n)\}^{1/2}$ with $\hat{B}_{\text{R}}(\hat{\bm{\beta}}_n) =n^{-1}\sum_{i=1}^{n}\hat{O}_i^2$, and
$$\hat{O}_i = \delta_i\Xi_i(\hat{\bm{\beta}}_n,X_i)
-\sum_{j=1}^{n}\frac{\delta_jY_i(X_j)\exp(\hat{\bm{\beta}}_n^\text{T}\bm{W}_i)}{\sum_{\ell=1}^{n}Y_\ell(X_j)\exp(\hat{\bm{\beta}}_n^\text{T}\bm{W}_\ell)} \Xi_i(\hat{\bm{\beta}}_n,X_j).$$
Under the \textit{simple randomization},
this test achieves an asymptotically valid size for testing $H_0$ even with a mis-specified working model \citep{lin}.   

Under the \textit{minimization} method, the denominator in the robust test statistic, $T_R$, needs to be adjusted \citep{ye,johnson2022validity}, which requires a consistent estimator for $\bm{\Sigma}_{\infty}(\bm{p}_0)$. Specifically,  
 for each $\bm{z} \in \mathcal{Z}$ and $j \in \{0,1\}$, denote by $n_{\bm{z}}$ the size of $\{i \in [n]: \bm{Z}_i = \bm{z}\}$, and by 
$\hat{E}_{\bm{z}j}$ and $\hat{V}_{\bm{z}j}$  the sample mean and variance  of $\{\hat{O}_i: \bm{Z}_i = \bm{z}, I_i=j; i \in [n],j=0,1 \}$. Further,
define
\begin{align}
    \label{def:B_adj}
    \hat{B}_{\text{Adj}}(\hat{\bm{\beta}}_n)
    = n^{-1} \sum_{\bm{z}\in\mathcal{Z}}n_{\bm{z}}\left( \hat{V}_{\bm{z}1}+\hat{V}_{\bm{z}0}\right)/2 +\hat{\bm{G}}^\text{T}\hat{\bm{\Sigma}}^B_n(\hat{\bm{p}}_n)\hat{\bm{G}},
\end{align}
where  $\hat{\bm{\Sigma}}^B_n(\hat{\bm{p}}_n)$ is the proposed Monte Carlo  estimator in Equation \eqref{def:boot_est}, and $\hat{\bm{G}} = (\hat{E}_{\bm{z}1}-\hat{E}_{\bm{z}0})/2$. 
Finally, we reject the null if
$|T_{\text{Adj}}| \geq \Phi^{-1}(\alpha/2)$, where $T_{\text{Adj}}=U_n(\hat{\bm{\beta}}_n)/\{\hat{B}_{\text{Adj}}(\hat{\bm{\beta}}_n)\}^{1/2}$. Under the \textit{minimization} design and  possible model mis-specification,
by \cite{ye,johnson2022validity}, under conditions C1-C3 in \cite{ye}, the adjusted test is asymptotically valid, that is, as $\min\{n,B\} \to \infty$,
$\mathbb{P}\{|T_{\text{Adj}}|>\Phi^{-1}(\alpha/2)\}\to  \alpha$, while the robust test may not be:
\begin{equation}
\label{def:R_size}
\lim_{n \to \infty} \mathbb{P}\{|T_\text{R}|>\Phi^{-1}(\alpha/2)\} =  2\Phi\left[-\Phi^{-1}\left(\alpha/2\right)\left\{\frac{\psi+\bm{G}^{\text{T}}\text{diag}(\bm{p}_0)\bm{G}}{\psi+\bm{G}^{\text{T}}\bm{\Sigma}_\infty(\bm{p}_0) \bm{G}}\right\}^{1/2}\right],
\end{equation}
where $\psi$ and $\bm{G}$ are introduced in \cite{ye} and also defined in Appendix \ref{app survival}. Recall that $\text{diag}(\bm{p}_0)$ is the limiting covariance of $n^{-1/2}\bm{S}_n$ under the simple randomization. 
If 
$\bm{G}^{\text{T}}\text{diag}(\bm{p}_0)\bm{G} < \bm{G}^{\text{T}}\bm{\Sigma}_\infty(\bm{p}_0) \bm{G}$, which could happen if 
$\bm{\Sigma}_\infty(\bm{p}_0)$ is not dominated by $\text{diag}(\bm{p}_0)$, then the robust test  $T_\text{R}$ has an asymptotically inflated size under the minimization method, which we demonstrate in Subsection \ref{subsec:inflation}.

\begin{remark}
In \citet[Equation (20)]{ye}, it is shown that $n^{-1} \sum_{\bm{z}\in\mathcal{Z}}n_{\bm{z}}\left( \hat{V}_{\bm{z}1}+\hat{V}_{\bm{z}0}\right)/2$ and 
$\hat{\bm{G}}$ are consistent for $\psi$ and $\bm{G}$ respectively. In this work, we propose a consistent estimator $\hat{\bm{\Sigma}}^B_n(\hat{\bm{p}}_n)$ for $\bm{\Sigma}_{\infty}(\bm{p}_0)$. As a result, by \citet[Theorem 1]{johnson2022validity}, we have that $\hat{B}_{\text{Adj}}(\hat{\bm{\beta}}_n)$ is a consistent estimator for the asymptotic variance of $U_n(\hat{\bm{\beta}}_n)$.
\end{remark}

\begin{remark}
The adjusted robust test $T_{\text{Adj}}$ is proposed for other covariate-adaptive designs in \cite{ye}, which is extended to the minimization method by \cite{johnson2022validity}. Note that  \cite{johnson2022validity} proposes a consistent estimator for $\bm{\Sigma}_{\infty}(\bm{p}_0)$ 
when each stratum has the same prevalence, that is,
 $\bm{p}_0 = (1/m,\ldots,1/m)$ (in particular, $\bm{p}_0$ is known in this case). Our main contribution is to propose the consistent estimator  $\hat{\bm{\Sigma}}^B_n(\hat{\bm{p}}_n)$ in \eqref{def:boot_est} with a general, unknown $\bm{p}_0$.
\end{remark}

\subsection{Numerical study: size and power of the adjusted test} \label{subsec:relevantsim}
In this simulation study, we compare four tests in terms of size and power with possibly mis-specified working models, including the adjusted robust test $T_{\text{Adj}}$, robust test $T_{\text{R}}$ without adjustment, unstratified log-rank test $T_{\text{L}}$ and  stratified log-rank test $T_{\text{SL}}$. The first two are introduced in the previous subsection;
$T_{\text{L}}$ corresponds to $T_{\text{R}}$ with covariates $\bm{W}_i = 0$, while $T_{\text{SL}}$ to $T_{\text{R}}$ with $\bm{W}_i = \{\mathbb{1}(\bm{Z}_i = \bm{z}): \bm{z} \in \mathcal{Z}\}$, i.e., a weighted average of stratum specific log-rank test statistics. For $T_{\text{Adj}}$, we use the empirical probability mass function $\hat{\bm{p}}_n$ and $B=1,000$ for $\hat{\bm{\Sigma}}_{n}^{B}(\hat{\bm{p}}_n)$ in Equation \eqref{def:boot_est}.

For the minimization method, we set weight vector ${\omega}=(1/K,\ldots,1/K)$ and use $g_{q}(\cdot)$ in Equation \eqref{def:gfunction} with $q = 2/3$  in Step 3.  For each $k \in [K]$, $\bm{Z}^{k}$ has a discrete uniform distribution over $\{0,\ldots,m_k-1\}$, and $\bm{Z}^{1},\ldots,\bm{Z}^K$ are independent. We consider the following cases with $n \in \{200,500,750\}$, where $h_0=\log(2)/12$. We recall that $h(\cdot\mid {V},I)$ and $\tilde{h}(\cdot\mid \bm{V},I)$ are the true and working hazard rate function respectively, and that $C$ denotes the censoring time. In addition, 
denote by $\text{Unif}(\ell,u)$ (resp. $\varepsilon_1$) a random variable that has the uniform distribution over $(\ell,u)$ (resp. exponential distribution with rate $1$), which is independent of all other random variables.

\textbf{Case 1.} $m_1 = 2$ and $m_2 = 3$, i.e., $m=2\times 3$. Let ${W}_1 = \bm{Z}^1$, $\bm{W}_2 = \mathbb{1}(\bm{Z}^2 = 0)$,
$\bm{W}_3 = \mathbb{1}(\bm{Z}^2 = 1)$.  
$h(t\mid \bm{V},I) = h_0 \text{exp}(\theta I+1.5 \bm{W}_1-\bm{W}_2-0.5 \bm{W}_3)$  for $t \in \mathbb{R}$, and 
 $\tilde{h}(\cdot \mid \bm{V},I)=h(\cdot\mid \bm{V},I)$. 
 $C = \text{Unif}(20,40)$.

 \textbf{Case 2.} 
$m_1=m_2=2$ and $m_3 = 5$, i.e. $m=2\times 2\times 5$.  
Let $\bm{W}_1 = \bm{Z}^1$, and $\bm{W}_2$ has the 
$\mathcal{N}_1(0,1)$ distribution and is independent from $\bm{Z}$.
$h(t\mid \bm{V},I) = h_0\text{exp}(\theta I -1.5\bm{W}_1+ 2.5\bm{W}_2)$ for $t \in \mathbb{R}$, and $\tilde{h}(\cdot \mid \bm{V},I)=h(\cdot\mid \bm{V},I)$. 
$C = \text{Unif}(40,70)$.

\textbf{Case 3.} 
$m_1 = 2$ and $m_2 = 4$, i.e., $m=2\times 4$. Let $\bm{W}_2$ has the 
$\mathcal{N}_1(0,1)$ distribution; $\bm{W}_1 = \bm{Z}^1$ is independent from $\bm{W}_2$, and   $\bm{Z}^2 = \sum_{j=1}^{3} \mathbb{1}(\bm{W}_2 \geq \Phi^{-1}(j/4))$.
The  failure time  
$T = \text{exp}(\theta I+1.5\bm{W}_1+1.5(\bm{W}_2)^2)+\varepsilon_1$, and $\tilde{h}(t \mid \bm{V},I) =h_0\text{exp}(\theta I+\beta_1 \bm{W}_1+\beta_2\bm{W}_2)$ for $t \in \mathbb{R}$. 
$C =\text{Unif}(0,1)$.

\textbf{Case 4($\kappa$).}
$m_1 =2$, $m_2 = 2$, $m_3 = \kappa \in \{5,10\}$, i.e.,
$m = 2\times 2 \times \kappa$. Let $\bm{W}_1=\bm{Z}^1$ independent from $\bm{W}_2,\bm{W}_3$, which are i.i.d. $\mathcal{N}_1(0,1)$ random variables; $\bm{Z}^2 = \mathbb{1}(\bm{W}_2 \geq 0)$.  $h(t\mid \bm{V},I)=h_0\text{exp}(\theta I+2\bm{W}_1+2.5\bm{W}_3)$, and $\tilde{h}(t\mid \bm{V},I) = h_0\text{exp}(\theta I+\beta_1\bm{W}_3)$ for $t\in \mathbb{R}$. 
$C = \text{Unif}(40,70)$.

The mean censoring rate for Case 1-3 is respectively $20.3\%$, $37.3\%$, $18.9\%$, while for Case 4 with $\kappa =5$ and $10$, it is $28.4\%$ and $18.1\%$. Case $1$ and $2$ consider correctly specified working models, and Case 2 has a larger number of strata. Case 3 considers model mis-specification in terms of the proportional hazard assumption, while Case $4$ in terms of missing covariates. 

In Table \ref{table:size}, we report the sizes of the four tests under the above cases with $\theta=0$ for $n\in\{200,500,750\}$ at the $5\%$ nominal level based on $10^5$ repetitions. Table \ref{table:size} shows that the unstratified log-rank test $T_{\text{L}}$, which is widely used in clinical trials, tends to be conservative under the minimization method. Further, 
if  \textit{all} the stratification factors used in the randomization are included calculation of test statistic, 
the stratified log-rank test $T_{\text{SL}}$ maintains a correct test size, which is established in \cite{ye}.  The robust test $T_{\text{R}}$ under the minimization method tends to be conservative in the presence of model mis-specification. With the adjustment using the proposed covariance estimator, we find $T_{\text{Adj}}$ corrects the size of $T_{\text{R}}$ in the right direction and attains a size that is close to the nominal level for a moderate sample size.




In Figure \ref{fig:power1} and \ref{fig:power2}, we further report the power of the above tests corresponding to various $\theta$ in the above cases for $n = 500$ at the $5\%$ nominal level based on $2,000$ repetitions. Compared to the stratified log-rank test $T_{\text{SL}}$, the robust test $T_{\text{R}}$ and its adjusted version $T_{\text{Adj}}$ have a much larger power when an important covariate, in addition to those stratification factors, is included in the working model, such as in $\bm{W}_2$ in Case 2 and $\bm{W}_3$ in Case 4. When the robust test $T_{\text{R}}$ is conservative, such as in Case 4, the adjusted version $T_{\text{Adj}}$ is more powerful. These empirical results are in line with \citet[Corollary 1 and Theorem 2]{ye}.

\begin{table}[bt]
    \centering
    \caption{Empirical Test Sizes under $4$ cases. $T_{\text{L}}$: Log-rank test; $T_{\text{R}}$: Robust test; $T_{\text{SL}}$: Stratified Log-rank test; $T_{\text{Adj}}$: Adjusted robust test.
}
\label{table:size}
    \begin{tabular}{ccccccccccccccc}
    \headrow
        & \multicolumn{4}{c}{$n=200$}  &  &\multicolumn{4}{c}{$n=500$} & & \multicolumn{4}{c}{$n=750$}\\ 
                & $T_{\text{L}}$& $T_{\text{R}}$ &$T_{\text{SL}}$ & $T_{\text{Adj}}$ &
                & $T_{\text{L}}$&$T_{\text{R}}$ &$T_{\text{SL}}$ & $T_{\text{Adj}}$&
                & $T_{\text{L}}$&$T_{\text{R}}$ &$T_{\text{SL}}$ & $T_{\text{Adj}}$\\ \hline
    $\text{C} 1$ &1.8& 5.5 & 5.1 & 5.5& & 1.7& 5.2 & 5.0& 5.3& &1.6 & 5.1 & 4.9 & 5.2\\ 
    $\text{C} 2$ & 4.5&5.7& 4.9 & 5.8& &4.5&5.1&5.0& 5.3& &4.3&5.4&4.9& 5.1\\
    $\text{C} 3$ & 2.2 & 3.4 & 5.0 & 6.1& &2.0&2.9 &5.2& 5.5& &1.9&2.7&5.0& 5.3\\
    $\text{C} 4(\kappa=5)$ &4.2&2.5 &5.1 & 6.2& &3.9&2.1&5.0&5.4 & &3.9&2.0&5.0& 5.2\\
    $\text{C} 4(\kappa=10)$ &4.2&2.3 & 5.0 &  8.4& &4.0&2.0&5.2& 5.4& &4.1&1.9&4.9& 5.1\\
    \hline
\end{tabular}
\end{table}

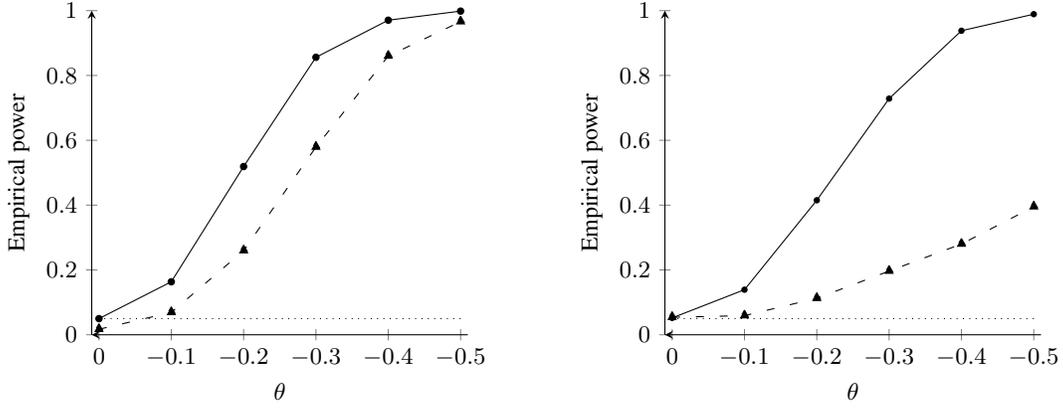
\begin{figure}[t!]
     \caption{Empirical power for Case 1 (left) and Case 2 (right). The solid line with circle marker is for the adjusted robust test, the dashed dotted line with diamond marker  for the conventional robust test, the loosely dashed line with triangle marker for the stratified log-rank test, and the densely dashed with square marker  for the log-rank test. The horizontal dotted line indicates the $5\%$ nominal level. \textbf{Visibly close lines are combined}: the robust test and the adjusted robust test are combined, and so are the log-rank test and the stratified log-rank test. }
    \centering
    \begin{tikzpicture}
        \begin{axis}[
            scale only axis,
            axis lines = left,
            ylabel = {\(\text{Empirical power}\)},
            ymin=0, ymax=1,
            xlabel = {\(\theta\)},
            xmax = 0.01,xmin = -0.51,
            x dir=reverse,
            legend pos=north west
        ]
        \addplot[black,mark=triangle*, mark size=2.1pt,loosely dashed]coordinates{(0,0.017)(-.1,0.069)(-.2,0.2595)(-.3,0.579)(-.4,0.86)(-.5,0.966)};
        \addplot [black,mark=*,mark size=1.2pt] coordinates{(0,0.0498)(-0.1,0.1635)(-0.2,0.519)(-0.3,0.856)(-0.4,0.9705)(-0.5, 0.9985)};
        \addplot [domain=-0.5:0,dotted]{0.05};
        \end{axis}

    \begin{axis}[
            scale only axis,
            axis lines = left,
            ylabel = {\(\text{Empirical power}\)},
            ymin=0, ymax=1,
            xlabel = {\(\theta\)},
            xmax = 0.01,xmin = -0.51,
            x dir=reverse,
            xshift=3in,
            legend pos=north west
        ]
        \addplot [black,mark=*,mark size=1pt] coordinates{(0,0.052)(-.1,0.139)(-.2,0.415)(-.3,0.7285)(-.4,0.938)(-.5,0.989)};
        \addplot [black, mark = triangle*, mark size =2.1pt,loosely dashed] coordinates{(0,0.054)(-.1,0.059)(-.2,0.1125)(-.3,0.1965)(-.4,0.28)(-.5,0.395)};
        \addplot [domain=-0.5:0,dotted]{0.05};
    \end{axis}
    \end{tikzpicture}
    \label{fig:power1}
\end{figure}

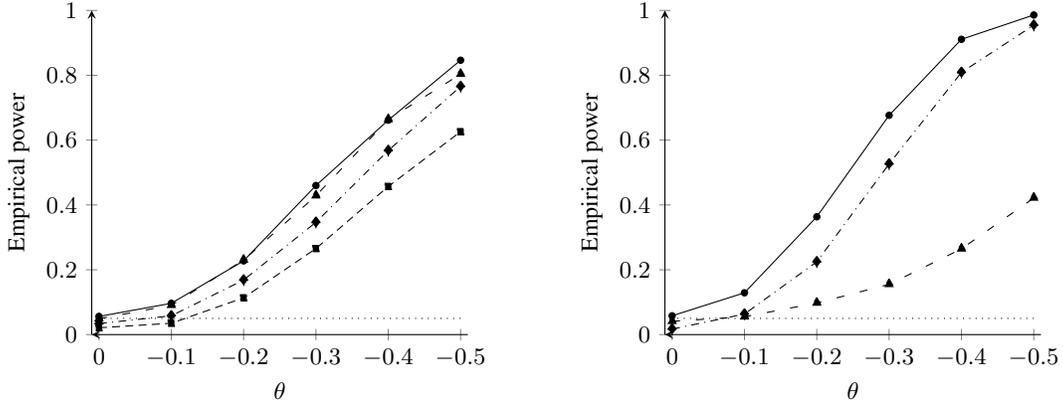
\begin{figure}
\caption{Empirical power of the four tests for Case 3 (left) and Case 4 with $\kappa=10$ (right). See explanations for each line in the caption of Figure \ref{fig:power1}. Lines for the  log-rank test and the stratified log-rank test are combined for Case 4.
}
    \begin{tikzpicture} 
    \begin{axis}[
            scale only axis,
            axis lines = left,
            ylabel = {\(\text{Empirical power}\)},
            ymin=0, ymax=1,
            xlabel = {\(\theta\)},
            xmax = 0.01,xmin = -0.51,
            x dir=reverse,
            legend pos=north west
        ]
        \addplot[black,mark=square*, mark size=1.2pt,densely dashed]coordinates{(0,0.021)(-.1,0.035)(-.2,0.1135)(-.3,0.2655)(-.4,0.457)(-.5,0.626)};
        \addplot [black,mark=diamond*,mark size=2.1pt,dashdotted] coordinates{(0,0.034)(-.1,0.058)(-.2,0.1685)(-.3,0.3465)(-.4,0.568)(-.5,0.7655)};
        \addplot [black, mark=triangle*, mark size =2.1pt,loosely dashed] coordinates{(0,0.049)(-.1,0.0905)(-.2,0.23)(-.3,.4285)(-.4,0.6635)(-.5,0.8035)};
         \addplot [black,mark=*,mark size=1.2pt] coordinates{(0,0.056)(-.1,0.0965)(-.2,0.2275)(-.3, 0.46)(-.4,0.6615789)(-.5,0.8465)};
        \addplot [domain=-0.5:0,dotted]{0.05};
    \end{axis}
    
    \begin{axis}[
            scale only axis,
            axis lines = left,
            ylabel = {\(\text{Empirical power}\)},
            ymin=0, ymax=1,
            xlabel = {\(\theta\)},
            xmax = 0.01,xmin = -0.51,
            x dir=reverse,
            xshift=3in,
            legend pos=north west
        ]
      \addplot[black, mark = triangle*, mark size =2.1pt,loosely dashed]coordinates{(0,0.041)(-.1,0.055)(-.2,0.0975)(-.3,0.155)(-.4,0.2645)(-.5,0.4215)};
        \addplot [black,mark=diamond*,mark size=2.1pt,dashdotted] coordinates{(0,0.017)(-.1, 0.0635)(-.2,0.2245)(-.3,0.5265)(-.4,0.809)(-.5,0.9545)};
         \addplot [black,mark=*,mark size=1.2pt] coordinates{(0,0.058)(-.1,0.129)(-.2,0.3635)(-.3,0.6765)(-.4,0.911)(-.5,0.986)};
        \addplot [domain=-0.5:0,dotted]{0.05};
    \end{axis}
    \end{tikzpicture}
\label{fig:power2}
\end{figure}

\subsection{Numerical study: size inflation under the minimization method}\label{subsec:inflation}

Under the minimization method, the asymptotic size of the robust test $T_{\text{R}}$ \textit{without} adjustment is given in Equation \eqref{def:R_size}. 
In Subsection \ref{subsec:dominate}, we show  by simulation that  $\bm{\Sigma}_\infty(\bm{p}_0)$ is not dominated by $\text{diag}(\bm{p}_0)$ under some parameter configurations, and thus $T_{\text{R}}$ may have an inflated size, which we now demonstrate.

Specifically, we consider two stratification factors each with two levels, i.e., $m = 2\times 2$, and $\bm{p}_0 = (1/4,\ldots,1/4)$. For the minimization method, we set weight vector $\bm{\omega}=(1/2,1/2)$ and use $g_{q}(\cdot)$ in Equation \eqref{def:gfunction} with $q = 0.4$  in Step 3.
We consider two models for the failure time. In the first model termed ``Exponential", the true hazard function $h(t\mid \bm{V},I) = 0.1\times\mathbb{1}({Z}^1=\bm{Z}^2)+2\times\mathbb{1}(\bm{Z}^1\neq\bm{Z}^2)$ for $t \in \mathbb{R}$; in the second model termed ``Gamma",
$T = \Gamma_1\mathbb{1}(\bm{Z}^1=\bm{Z}^2)+\Gamma_2\mathbb{1}(\bm{Z}^1\neq\bm{Z}^2)$, where  $\Gamma_1$ (resp. $\Gamma_2$) has the Gamma distribution with the (shape, rate) parameter $(10,2)$ (resp. $(40,2)$). We assume no censoring, and use the
working hazard function $h_0(t)\text{exp}\{\theta I+\beta_1\mathbb{1}(\bm{Z}^1=1)\}$, that is, partial stratification on $\bm{Z}^1$.

As discussed in Subsection \ref{subsec:dominate}, $\bm{\Sigma}_{\infty}(\bm{p}_0)$ has a structure given in Equation \eqref{def:specialcov}, and assuming $\bm{p}_0$ is known,  
we estimate  $v(g_{0.4})$ in Equation  \eqref{def:specialcov} 
by the first entry of $\hat{\bm{\Sigma}}^{B'}_{n'}(\bm{p}_0)$ in Equation  \eqref{def:boot_est} with $B' = 4\times 10^5$  and $n'=10^4$. Further, we estimate $\psi$ and ${G}$ in Equation \eqref{def:R_size} by $\hat{\psi} = n^{-1} \sum_{\bm{z}\in\mathcal{Z}}n_{\bm{z}}(\hat{V}_{\bm{z}1}+\hat{V}_{\bm{z}0})/2$ and $\hat{\bm{G}}$, defined in Equation \eqref{def:B_adj}, with a sample size of $10^4$ and $4\times 10^5$ Monte Carlo repetitions. Based on these estimates, i.e., $\bm{\Sigma}_{n'}^{B'}(\bm{p}_0), \hat{\psi}, \hat{\bm{G}}$, we calculate the theoretical size by \eqref{def:R_size}, which is   $5.29\%$ and $5.32\%$ for the Exponential and Gamma model respectively.

Further, we estimate the actual rejection probability under $H_0$ for $n \in\{200,500, 2\times10^3, 5\times10^3, 10^4\}$ at the $5\%$ nominal level with $4\times 10^5$ repetitions. In Table \ref{tab3}, we report the results, together with the standard deviation for these estimates. 
For a moderate sample size $n$, the unadjusted robust test has a $1\%$ absolute size inflation when one of the stratification factors used in randomization is not included in the stratified test. As $n$ increases, the empirical sizes approach the theoretical limits, which are larger than the nominal level, 
though the magnitude of inflation is smaller.

\begin{table}[bt]
\caption{
Theoretical sizes for the Exponential model and Gamma model are respectively $5.29\%$ and $5.32\%$. For all estimates, the standard deviation is around $0.04\%$.}
\label{tab3}
\centering
\begin{tabular}{cccccccccc}
\headrow
    &$n=200$ && $n=500$ & &$n=2,000$  &  
    & $n=5,000$ & 
    & $n=10,000$ 
    \\
Exponential &$6.25$ && $5.70$&&$5.43$ &&$5.38$&&$5.25$ \\  
Gamma&$6.21$&&$5.74$&& $5.47$&& $5.31$  &&$5.37$ \\
\hline
\end{tabular}
\end{table}



\section{Discussions}
In this work, we focus on the minimization method and propose a consistent Monte Carlo estimator for the limiting covariance matrix of within-stratum imbalances. Further, we conduct simulation studies to evaluate the performance when the estimator is used to adjust robust  tests in survival analysis.

There are several directions we plan to explore in the future. First, in this work, we consider balanced allocation under the minimization method, where the desired ratio for the treatment and control groups  is $1:1$ for each stratification factor; it is of practical importance to study minimization methods under an unbalanced allocation. Second, we assume there are only two treatment groups. A theory for multiple treatment groups was also developed in \cite{hu}.  It would be interesting  to develop consistent estimators and more importantly adjusted robust tests under this setup. Finally, in this work, we assume the number of strata $m$ is fixed when $n$ tends to infinity in the asymptotic analysis. Since the minimization method is commonly applied when $n$ is not much larger than $m$, it is important to consider asymptotics in which $m$ also diverges with $n$.

\section*{Acknowledgements}
The authors thank the associate editor and the two referees for careful reading of the manuscript and for constructive comments which led to significant improvement of the paper. This research was supported by the Natural Sciences and Engineering Research Council of Canada (NSERC). This research was also supported, in part, by Compute Canada (\texttt{www.computecanada.ca}).


\printendnotes
\bibliography{sample}

\begin{thebibliography}{33}
\expandafter\ifx\csname natexlab\endcsname\relax\def\natexlab#1{#1}\fi
\expandafter\ifx\csname url\endcsname\relax
  \def\url#1{\texttt{#1}}\fi
\expandafter\ifx\csname urlprefix\endcsname\relax\def\urlprefix{URL: }\fi

\bibitem[{Alphonso et~al.(2005)Alphonso, Tan, Utley, Cameron, Dussek,
  Lang-Lazdunski and Treasure}]{alphonso2005prospective}
Alphonso, N., Tan, C., Utley, M., Cameron, R., Dussek, J., Lang-Lazdunski, L.
  and Treasure, T. (2005) A prospective randomized controlled trial of suction
  versus non-suction to the under-water seal drains following lung resection.
\newblock \textit{European journal of cardio-thoracic surgery}, \textbf{27},
  391--394.

\bibitem[{Asmussen(2003)}]{asmussen}
Asmussen, S. (2003) \textit{Applied probability and queues}, vol.~2.
\newblock Springer.

\bibitem[{Baldi~Antognini and Zagoraiou(2011)}]{baldi}
Baldi~Antognini, A. and Zagoraiou, M. (2011) The covariate-adaptive biased coin
  design for balancing clinical trials in the presence of prognostic factors.
\newblock \textit{Biometrika}, \textbf{98}, 519--535.

\bibitem[{Breugom et~al.(2015)Breugom, Van~Gijn, Muller, Berglund, Van~den
  Broek, Fokstuen, Gelderblom, Kapiteijn, Leer, Marijnen
  et~al.}]{breugom2015adjuvant}
Breugom, A., Van~Gijn, W., Muller, E., Berglund, {\AA}., Van~den Broek, C.,
  Fokstuen, T., Gelderblom, H., Kapiteijn, E., Leer, J., Marijnen, C. et~al.
  (2015) Adjuvant chemotherapy for rectal cancer patients treated with
  preoperative (chemo) radiotherapy and total mesorectal excision: a dutch
  colorectal cancer group (dccg) randomized phase iii trial.
\newblock \textit{Annals of Oncology}, \textbf{26}, 696--701.

\bibitem[{Bugni et~al.(2018)Bugni, Canay and Shaikh}]{bugni2018inference}
Bugni, F.~A., Canay, I.~A. and Shaikh, A.~M. (2018) Inference under
  covariate-adaptive randomization.
\newblock \textit{Journal of the American Statistical Association},
  \textbf{113}, 1784--1796.

\bibitem[{Bugni et~al.(2019)Bugni, Canay and Shaikh}]{bugni2019inference}
--- (2019) Inference under covariate-adaptive randomization with multiple
  treatments.
\newblock \textit{Quantitative Economics}, \textbf{10}, 1747--1785.

\bibitem[{Bugni and Gao(2023)}]{bugni2023inference}
Bugni, F.~A. and Gao, M. (2023) Inference under covariate-adaptive
  randomization with imperfect compliance.
\newblock \textit{Journal of Econometrics}, \textbf{237}, 105497.

\bibitem[{Douc et~al.(2022)Douc, Jacob, Lee and Vats}]{douc2022solving}
Douc, R., Jacob, P.~E., Lee, A. and Vats, D. (2022) Solving the poisson
  equation using coupled markov chains.
\newblock \textit{arXiv preprint arXiv:2206.05691}.

\bibitem[{Efron(1971)}]{efron}
Efron, B. (1971) Forcing a sequential experiment to be balanced.
\newblock \textit{Biometrika}, \textbf{58}, 403--417.

\bibitem[{Hall and Heyde(2014)}]{hall}
Hall, P. and Heyde, C.~C. (2014) \textit{Martingale limit theory and its
  application}.
\newblock Academic press.

\bibitem[{Hu and Zhang(2020)}]{hu}
Hu, F. and Zhang, L.-X. (2020) On the theory of covariate-adaptive designs.
\newblock \textit{arXiv preprint arXiv:2004.02994}.

\bibitem[{Hunt et~al.(1992)Hunt, Varigos, Dienstl, Lechleitner, De~Backer,
  KORNITZER, CAIRNS, TURPIE, FRITZHANSEN, SKAGEN et~al.}]{hunt1992isis}
Hunt, D., Varigos, J., Dienstl, F., Lechleitner, P., De~Backer, G., KORNITZER,
  M., CAIRNS, J., TURPIE, A., FRITZHANSEN, P., SKAGEN, K. et~al. (1992) Isis-3:
  a randomised comparison of streptokinase vs tissue plasminogen activator vs
  antistreplase and of aspirin plus heparin vs aspirin along among 41 299 cases
  of suspected acute myocardial infarction.
\newblock \textit{Lancet}, \textbf{339}, 753--770.

\bibitem[{Johnson et~al.(2022)Johnson, Gekhtman and
  Kuznetsova}]{johnson2022validity}
Johnson, V.~P., Gekhtman, M. and Kuznetsova, O.~M. (2022) Validity of tests for
  time-to-event endpoints in studies with the pocock and simon
  covariate-adaptive randomization.
\newblock \textit{arXiv preprint arXiv:2202.03892}.

\bibitem[{Kosorok(2008)}]{kosorok}
Kosorok, M.~R. (2008) \textit{Introduction to empirical processes and
  semiparametric inference}.
\newblock Springer.

\bibitem[{Li et~al.(2021)Li, Ma, Qin and Hu}]{li}
Li, Y., Ma, W., Qin, Y. and Hu, F. (2021) Testing for treatment effect in
  covariate-adaptive randomized trials with generalized linear models and
  omitted covariates.
\newblock \textit{Statistical Methods in Medical Research}, \textbf{30},
  2148--2164.

\bibitem[{Lin and Wei(1989)}]{lin}
Lin, D.~Y. and Wei, L.-J. (1989) The robust inference for the cox proportional
  hazards model.
\newblock \textit{Journal of the American statistical Association},
  \textbf{84}, 1074--1078.

\bibitem[{Liu et~al.(2023)Liu, Tu and Ma}]{liu2023lasso}
Liu, H., Tu, F. and Ma, W. (2023) Lasso-adjusted treatment effect estimation
  under covariate-adaptive randomization.
\newblock \textit{Biometrika}, \textbf{110}, 431--447.

\bibitem[{Ma et~al.(2020)Ma, Qin, Li and Hu}]{ma2020statistical}
Ma, W., Qin, Y., Li, Y. and Hu, F. (2020) Statistical inference for
  covariate-adaptive randomization procedures.
\newblock \textit{Journal of the American Statistical Association},
  \textbf{115}, 1488--1497.

\bibitem[{Norris(1998)}]{norris1998markov}
Norris, J.~R. (1998) \textit{Markov chains}.
\newblock Cambridge university press.

\bibitem[{Van~der Ploeg et~al.(2010)Van~der Ploeg, Clemens, Corzo, Escolar,
  Florence, Groeneveld, Herson, Kishnani, Laforet, Lake
  et~al.}]{van2010randomized}
Van~der Ploeg, A.~T., Clemens, P.~R., Corzo, D., Escolar, D.~M., Florence, J.,
  Groeneveld, G.~J., Herson, S., Kishnani, P.~S., Laforet, P., Lake, S.~L.
  et~al. (2010) A randomized study of alglucosidase alfa in late-onset pompe's
  disease.
\newblock \textit{New England Journal of Medicine}, \textbf{362}, 1396--1406.

\bibitem[{Pocock and Simon(1975)}]{pocock}
Pocock, S.~J. and Simon, R. (1975) Sequential treatment assignment with
  balancing for prognostic factors in the controlled clinical trial.
\newblock \textit{Biometrics}, 103--115.

\bibitem[{Shao and Tu(2012)}]{shao2012jackknife}
Shao, J. and Tu, D. (2012) \textit{The jackknife and bootstrap}.
\newblock Springer Science \& Business Media.

\bibitem[{Shao et~al.(2010)Shao, Yu and Zhong}]{shao}
Shao, J., Yu, X. and Zhong, B. (2010) A theory for testing hypotheses under
  covariate-adaptive randomization.
\newblock \textit{Biometrika}, \textbf{97}, 347--360.

\bibitem[{Stott et~al.(2017)Stott, Rodondi, Kearney, Ford, Westendorp,
  Mooijaart, Sattar, Aubert, Aujesky, Bauer et~al.}]{stott2017thyroid}
Stott, D.~J., Rodondi, N., Kearney, P.~M., Ford, I., Westendorp, R.~G.,
  Mooijaart, S.~P., Sattar, N., Aubert, C.~E., Aujesky, D., Bauer, D.~C. et~al.
  (2017) Thyroid hormone therapy for older adults with subclinical
  hypothyroidism.
\newblock \textit{New England Journal of Medicine}, \textbf{376}, 2534--2544.

\bibitem[{Taves(1974)}]{taves}
Taves, D.~R. (1974) Minimization: a new method of assigning patients to
  treatment and control groups.
\newblock \textit{Clinical Pharmacology \& Therapeutics}, \textbf{15},
  443--453.

\bibitem[{Taves(2010)}]{taves2}
--- (2010) The use of minimization in clinical trials.
\newblock \textit{Contemporary clinical trials}, \textbf{31}, 180--184.

\bibitem[{Van~der Vaart(2000)}]{van}
Van~der Vaart, A.~W. (2000) \textit{Asymptotic statistics}, vol.~3.
\newblock Cambridge university press.

\bibitem[{Wei and Lachin(1988)}]{wei1988properties}
Wei, L. and Lachin, J.~M. (1988) Properties of the urn randomization in
  clinical trials.
\newblock \textit{Controlled clinical trials}, \textbf{9}, 345--364.

\bibitem[{Wei(1978)}]{wei1978adaptive}
Wei, L.-J. (1978) The adaptive biased coin design for sequential experiments.
\newblock \textit{The Annals of Statistics}, \textbf{6}, 92--100.

\bibitem[{Xu et~al.(2016)Xu, Proschan and Lee}]{xu2016validity}
Xu, Z., Proschan, M. and Lee, S. (2016) Validity and power considerations on
  hypothesis testing under minimization.
\newblock \textit{Statistics in medicine}, \textbf{35}, 2315--2327.

\bibitem[{Ye and Shao(2020)}]{ye}
Ye, T. and Shao, J. (2020) Robust tests for treatment effect in survival
  analysis under covariate-adaptive randomization.
\newblock \textit{Journal of the Royal Statistical Society: Series B
  (Statistical Methodology)}, \textbf{82}, 1301--1323.

\bibitem[{Ye et~al.(2022)Ye, Yi and Shao}]{ye2022inference}
Ye, T., Yi, Y. and Shao, J. (2022) Inference on the average treatment effect
  under minimization and other covariate-adaptive randomization methods.
\newblock \textit{Biometrika}, \textbf{109}, 33--47.

\bibitem[{Zelen(1974)}]{zelen1974randomization}
Zelen, M. (1974) The randomization and stratification of patients to clinical
  trials.
\newblock \textit{Journal of Chronic Diseases}, 365--75.

\end{thebibliography}



\begin{appendix}

\section{Proofs related to the minimization method} \label{app}

\subsection{Previous results} \label{app:previous results}
First, we recall some  results from \cite{hu} that are needed in the proof.

Recall, for $i \in [n]$, $k\in [K]$ and $\bm{z}^{k} \in \mathcal{Z}^{k}$, the definition of $\bm{M}_{i}(k,\bm{z}^k)$ in Equation \eqref{marginalimb}, and  the definition of  within-stratum imbalances $\bm{S}_n$ after the $n$-th assignment in Equation \eqref{stratum imbalance}. 
Further, define 
$$\bm{\Lambda}_i = (\bm{\Lambda}_i(\bm{z}):\bm{z}\in \mathcal{Z} ), \quad \text{ where } \bm{\Lambda}_i(\bm{z}) =\sum_{k=1}^{K}\bm{\omega}^{k}\bm{M}_ i\left(k;\bm{z}^k\right).$$
Recall the definition of the potential imbalances $\text{Imb}_i^{(1)}$ and $\text{Imb}_i^{(0)}$ in Equation \eqref{Imbalance measure}. Since $(x+1)^2-(x-1)^2=4x$, we have that
$$\text{Imb}_i^{(1)}-\text{Imb}_i^{(0)} = 4 \sum_{k=1}^{K} \bm{\omega}^{k}\left\{\bm{M}_{i-1}\left(k,\bm{Z}_i^k\right)\right\} = 4\bm{\Lambda}_{i-1}(\bm{Z}_i).$$ 

Next, we introduce additional notations. Let $\mathcal{S}$ be a countable state space. For any function $f : \mathcal{S} \to \mathbb{R}$, probability mass function (pmf) $\pi$ on $\mathcal{S}$, and transition probabilities $\bm{P} = (\bm{P}_{x,y}:x,y \in \mathcal{S})$, denote by $\pi(f) = \sum_{x \in \mathcal{S}} \pi(x) f(x)$ and $\bm{P} f (x) = \sum_{y\in \mathcal{S}} \bm{P}_{x,y}f(y)$ whenever the summation is well defined.

\begin{mylem} \label{lemma:relevant properties1}
   Consider the minimization method in Subsection \ref{subsec:designnotationsandprocedure}. 
   \begin{enumerate}[label=(\roman*)]
        \item The process $(\bm{\Lambda}_n)_{n\geq 1}$ is an irreducible positive recurrent Markov Chain on a countable state space $\mathcal{S} \subset \mathbb{R}^{m}$,
        where $0 \in \mathcal{S}$, and if $\bm{\lambda}\in \mathcal{S}$, then $-\bm{\lambda}\in \mathcal{S}$. Further, $(\bm{\Lambda}_n)_{n\geq 1}$ and $(-\bm{\Lambda}_n)_{n\geq 1}$ have the same distribution.
        \item  For each $\bm{z}\in \mathcal{Z}, r>0$, $\sup_{n\geq 1} \mathbb{E}|n^{-1/2}\bm{S}_n(\bm{z})|^r<\infty$.
    \end{enumerate}
\end{mylem}

\begin{proof}
    For claims (i) and (ii), see Proposition 3.1(iii) and Theorem 3.2(iv) in \cite{hu} respectively.
\end{proof}

Let $\mathcal{F}_0$ be the trivial $\sigma$-algebra, and for $i \geq 1$, $\mathcal{F}_i =\sigma(\mathcal{F}_{i-1},I_{i-1},\bm{Z}_i)$ be the $\sigma$-algebra generated by
$\mathcal{F}_{i-1}$ and $I_{i-1},\bm{Z}_i$, i.e., $\mathcal{F}_i$ is
the information set for the $i$-th assignment. Recall Step $3$ of the minimization method in Subsection \ref{subsec:designnotationsandprocedure}, by definition, 
$g(\text{Imb}_i^{(1)}-\text{Imb}_i^{(0)}) = P(I_i=1\mid \mathcal{F}_{i-1})$. For each $\bm{z} \in \mathcal{Z}$ and $\bm{\lambda} \in \mathcal{S}$, for simplicity define 
\begin{align}
\label{def:g_bar}    
\bar{g}_{\bm{z}}(\bm{\lambda}) = g\{4\bm{\lambda}(\bm{z})\} - g\{-4\bm{\lambda}(\bm{z})\}.
\end{align}
Denote by $\bm{P}_{\bm{\Lambda}}$ the transition probability matrix of $(\bm{\Lambda}_n)_{n\geq 1}$ with the ($\bm{\lambda},\bm{\lambda}'$)-th component as $\bm{P}_{\bm{\Lambda}}(\bm{\lambda},\bm{\lambda}') = \mathbb{P}(\bm{\Lambda}_n= \bm{\lambda}'\mid \bm{\Lambda}_{n-1}=\bm{\lambda})$.

\begin{mylem}  \label{lemma:relevant properties2}
Consider the minimization method in Subsection \ref{subsec:designnotationsandprocedure}. 
\begin{enumerate}[label=(\roman*)]
    \item 
    Consider the Poisson equation 
    $$h -\bm{P}_{\bm{\Lambda}} h =\bar{g}_{\bm{z}},$$ 
    where $h: \mathcal{S} \to \mathbb{R}$ is to be solved. It
    admits a unique solution $\hat{g}_{\bm{z}}: \mathcal{S} \to \mathbb{R}$, for which $\hat{g}_{\bm{z}}(\bm{0}) = 0$.
    \item For each $\bm{z} \in \mathcal{Z}, r>0$, $\sup_n \mathbb{E}|\hat{g}_{\bm{z}}(\bm{\Lambda}_n)|^r<\infty$. 
    \item  For each $\bm{z} \in \mathcal{Z}$ and $i \geq 1$, the sequence 
    \begin{equation}
         \Delta M_{i,z}= \bm{S}_i(\bm{z})-\bm{S}_{i-1}(\bm{z})+\bm{p}_0(\bm{z})\{\hat{g}_{\bm{z}}(\bm{\Lambda}_i)-\hat{g}_{\bm{z}}(\bm{\Lambda}_{i-1})\}
         \label{deltaMiz}
    \end{equation}
    is a martingale difference sequence with respect to $\left\{\mathcal{F}_i:i\geq 1\right\}$; that is, $\mathbb{E}(\Delta M_{i,\bm{z}}\mid \mathcal{F}_{i-1}) = 0$.
\end{enumerate}
\end{mylem}

\begin{proof}
For the existence of the Poisson equation solution in claim (i), see \citet[Equation (6.11)]{hu}. Since $\hat{g}$ is $pi$-integrable and $\hat{g}_{\bm{z}}(\bm{0})=0$, by \citet[Proposition 7.1]{asmussen}, the uniqueness is guaranteed.
For claim (ii) and (iii), see  \citet[Equation (6.5),(6.11)]{hu},
and \citet[Equation (6.14)]{hu} respectively. 
\end{proof}

\subsection{Auxiliary lemmas} \label{app:lemmas}

In this subsection, we establish several results.
We start with a structural property of Poisson equation solutions associated with a general symmetric Markov chain.

\begin{mylem}\label{lemma:odd function}
    Let $\{X_n: n \geq 1\}$ be an irreducible positive recurrent Markov Chain on a countable state space $\mathcal{S}_X\subset \mathbb{R}^m$, with transition probabilities $\bm{P} = \{\bm{P}_{x,x'}: x, x' \in \mathcal{S}_X\}$. Denote by $\bm{\pi}$ the unique  invariant distribution. Assume that $\bm{0}\in \mathcal{S}_X$ and if $x\in \mathcal{S}_X$, then $-x\in \mathcal{S}_X$. Further assume $\bm{P}_{x,x'} = \bm{P}_{-x,-x'}$ for any $x,x'\in \mathcal{S}_X$. Let $f:\mathcal{S}_X \to \mathbb{R}$ be an $\bm{\pi}$ integrable function such that $\bm{\pi}(f)=0$. 
    
    \begin{itemize}
        \item[(i)] The Poisson equation $h-\mathbb{P} h = f$ admits a unique solution $h^*:\mathcal{S} \to \mathbb{R}$ with $h^*(\bm{0}) = 0$.
        \item[(ii)] If moreover, $f(x)=-f(-x)$, for all $x\in \mathcal{S}_X$, then  $h^*(x) = -h^*(-x)$.
    \end{itemize}
\end{mylem}

\begin{proof}
    Denote $\mathbb{P}_{x}$ as the probability measure, under which $\{X_n: n \geq 0\}$ is a Markov chain with initial state $x$, and $\mathbb{E}_{x}$ the corresponding expectation.
    By \citet[Proposition 7.1]{asmussen}, if we let $\tau = \inf \left\{n\geq 1:X_n=0\right\}$,
    then the unique  solution with the property of $h^*(\bm{0}) = 0$ is of the form 
    $$h^*(x) = \mathbb{E}_{x}\left\{\sum_{i=0}^{\tau-1}f(X_i)\right\},\quad \text{ for any } x \in \mathcal{S}_X,
    $$
which is the claim in (i).

From the proof of \citet[Proposition 7.1]{asmussen},
for any $x \in \mathcal{S}_X$, $\mathbb{E}_{x}\left\{\sum_{i=0}^{\tau-1} |f(X_i)|\right\} < \infty$. Then by dominated convergence theorem and since $f$ is an odd function,

    \begin{align}
    \begin{split}
                    h^*(x) &=\sum_{n=1}^{\infty} \sum_{x_0\in \mathcal{S}_X}\ldots \sum_{x_n\in  \mathcal{S}_X} \mathbb{P}_{x}(\tau=n,X_0=x_0,\ldots,X_n=x_n)\left\{\sum_{i=0}^{n-1}f(x_i)\right\},
        \\ 
            h^*(-x) &=\sum_{n=1}^{\infty} \sum_{x_0\in  \mathcal{S}}\ldots \sum_{x_n\in  \mathcal{S}}\mathbb{P}_{-x}(\tau=n,X_0=-x_0,\ldots,X_n=-x_n)\left\{-\sum_{i=0}^{n-1}f(x_i)\right\}.
                \end{split}
\label{eq:poisson}
    \end{align}

Fix some arbitrary $n \geq 1$, and $x_0,\ldots,x_n \in \mathcal{S}_X$. If the following holds
\begin{align}
    x_j \neq 0 \;\;\text{ for } 1 \leq j \leq n-1, \quad \text{ and }\;\;x_n = 0, \label{tau_n_event}
\end{align}
then we have
\begin{align*}
  &\{\tau=n,X_0=x_0,\ldots,X_n=x_n\} =   \{  X_0=x_0,\ldots,X_n=x_n\},\\
  &\{\tau=n,X_0=-x_0,\ldots,X_n=-x_n\}=\{X_0=-x_0,\ldots,X_n=-x_n\}.
\end{align*}
Since $\bm{P}_{x,x'}=\bm{P}_{-x,-x'}$ for any $x,x'\in \mathcal{S}_X$, if condition \eqref{tau_n_event} holds, then
$$\mathbb{P}_{x}(\tau=n,X=x_0,\ldots,X_n=x_n)=\mathbb{P}_{-x}(\tau=n, X=-x_0,\ldots,X_n=-x_n).$$
On the other, if condition \eqref{tau_n_event} does not hold, then 
$$\{\tau=n,X=x_0,\ldots,X_n=x_n\} = \emptyset = \{\tau=n, X=-x_0,\ldots,X_n=-x_n\}.$$

Combining the two cases,  we complete the proof  due to Eqs.\eqref{eq:poisson}.
\end{proof}

Recall the Markov chain $(\bm{\Lambda}_n)_{n\geq 1}$ with the state space $\mathcal{S}$ and the transition probabilities  $\bm{P}_{\bm{\Lambda}}$ in Lemma \ref{lemma:relevant properties1}(i), and the 
    unique Poisson equation solution $\hat{g}_{\bm{z}}(\cdot)$ for each $\bm{z} \in \mathcal{Z}$ in Lemma \ref{lemma:relevant properties2}(i).

\begin{mycor}\label{cor:oddfunction}
  For each $\bm{z} \in \mathcal{Z}$ and $\bm{\lambda} \in \mathcal{S}$, 
    $\hat{g}_{\bm{z}}(\bm{\lambda}) = - \hat{g}_{\bm{z}}(-\bm{\lambda})$.
\end{mycor}

\begin{proof}
Fix some $\bm{z} \in \mathcal{Z}$.
    By definition in Equation \eqref{def:g_bar},  $\bar{g}_{\bm{z}}(-\bm{\lambda}) = -\bar{g}_{\bm{z}}(\bm{\lambda})$  for each $ \bm{\lambda}\in \mathcal{S}$.   Then the proof is complete due to Lemma \ref{lemma:odd function}.
\end{proof}

\subsection{Local asymptotic normality} \label{app:lam}
For any probability mass function (pmf) $\bm{p} \in \bm{\Delta}_m=\{\bm{p}\in \mathbb{R}^m: \sum_{\bm{z} \in \mathcal{Z}}\bm{p}(\bm{z})=1,\; \bm{p}(\bm{z}) \geq 0 \text{ for each } \bm{z} \in \mathcal{Z}\}$, recall the definition in Equation \eqref{ell} for the log-likelihood ratio, $\ell_n(\bm{p};\bm{p}_0)$, of $\bm{p}$ against $\bm{p}_0$ for the sequence $\bm{Z}_1,\ldots, \bm{Z}_n$. We use the notation $o_{\bm{p}}(1)$ to indicate a sequence of random variables that converges to $0$ in probability as $n \to \infty$, when the pmf of $\bm{Z}$ is $\bm{p} \in \bm{\Delta}_m$.

\begin{proof}[Proof of Lemma \ref{lemma:LAN}]
For each pmf $\bm{p} \in \bm{\Delta}_m$, the degree of freedom is $m-1$; thus, we fix some stratum $\bm{z}^* \in \mathcal{Z}$, and denote by 
$\mathcal{Z}^{-*} = \mathcal{Z}\setminus \{\bm{z}^*\}$. The pmf $\bm{p} \in \bm{\Delta}_m$ of the generic random vector $\bm{Z}$  can be written as follows:
    $$\bm{p}(\bm{Z})=\left\{\prod_{\bm{z}\in \mathcal{Z}^{-*} } \bm{p}(\bm{z})^{\mathbb{1}(\bm{Z}=\bm{z})}\right\}\left\{1-\sum_{\bm{z}\in \mathcal{Z}^{-*} }\bm{p}(\bm{z})\right\}^{\mathbb{1}(\bm{Z} = \bm{z}^*)}.$$  By \citet[Lemma 7.6]{van} , the probabilistic model for $\bm{Z}$, parameterized by $\bm{p} \in \bm{\Delta}_m$, is differentiable at $\bm{p}_0$ in quadratic mean with derivative 
    $$
    \left\{ {\mathbb{1}(\bm{Z} = \bm{z})}/{\bm{p}_0(\bm{z})} - {\mathbb{1}(\bm{Z} = \bm{z}^*)}/{\bm{p}_0(\bm{z}^*)}
    \;\;:\;\;\bm{z} \in \mathcal{Z}^{-*}  
    \right\},
    $$
    and the information matrix is
    $$
     \left\{ {\mathbb{1}(\bm{z} = \bm{z}')}/{\bm{p}_0(\bm{z})} + 1/{\bm{p}_0(\bm{z}^*)}
    \;\;:\;\;\bm{z}, \bm{z}'
    \in \mathcal{Z}^{-*}  
    \right\}.
    $$
    Thus by \citet[Theorem 7.2]{van}, we have
    \begin{align*}          
        \ell_n(\bm{p}_n;\bm{p}_0) = &n^{-1/2} \sum_{i=1}^{n}\sum_{\bm{z} \in \mathcal{Z}^{-*}}  \bm{\delta}(\bm{z}) \left\{{\mathbb{1}(\bm{Z}_i = \bm{z})}/{\bm{p}_0(\bm{z})} - {\mathbb{1}(\bm{Z}_i = \bm{z}^*)}/{\bm{p}_0(\bm{z}^*)} \right\}
        \\&- 2^{-1} \sum_{\bm{z}, \bm{z}'\in \mathcal{Z}^{-*} } \bm{\delta}(\bm{z})\bm{\delta}(\bm{z}')\left\{ {\mathbb{1}(\bm{z} = \bm{z}')}/{\bm{p}_0(\bm{z})} + {1}/{\bm{p}_0(\bm{z}^*)}\right\}+o_{\bm{p}_0}(1).
    \end{align*}
Then the proof is complete since $\bm{\delta} = \lim_{n \to \infty} n^{1/2}(\bm{p}_n-\bm{p}_0)$, and as a result  $\bm{\delta}(\bm{z}^*) = -\sum_{\bm{z}\in \mathcal{Z}^{-*} }\bm{\delta}(\bm{z})$.
\end{proof}

Consider the minimization method in Subsection \ref{subsec:designnotationsandprocedure}. As discussed there, the  distribution of the \textit{joint} sequence $\{\bm{Z}_i,I_i: i \in [n]\}$ is determined by the probability mass function (pmf) $\bm{p} \in \bm{\Delta}_m$ of the generic random vector $\bm{Z}$. Thus, when $\bm{Z}$ has pmf $\bm{p}$, we denote the probability and expectation as $\mathbb{P}_{\bm{p}}$ and $\mathbb{E}_{\bm{p}}$ respectively. The next lemma shows that $\exp\{\ell_n(\bm{p}; \bm{p}_0)\}$ in \eqref{ell} is, in fact, the likelihood ratio of $\mathbb{P}_{\bm{p}}$ against $\mathbb{P}_{\bm{p}_0}$  for this \textit{joint} sequence.

\begin{mylem}[change of measure] \label{lemma:changeofmeasure}
Consider the minimization method in Subsection \ref{subsec:designnotationsandprocedure}, and let $\bm{p}\in \bm{\Delta}_m$ be any probability mass function.  For any  non-negative function $\zeta:\mathcal{Z}^n\times \{0,1\}^n\to[0,\infty)$,
    \begin{equation*}
       \mathbb{E}_{\bm{p}}\left\{\zeta(\bm{Z}_1,\ldots,\bm{Z}_n,I_1,\ldots,I_n)\right\} = \mathbb{E}_{\bm{p}_0}\left[\exp\{\ell_n(\bm{p};\bm{p}_0)\} \zeta(\bm{Z}_1,\ldots,\bm{Z}_n,I_1,\ldots,I_n)\right].
    \end{equation*}
\end{mylem}

\begin{proof}
    Under the minimization method in Subsection \ref{subsec:designnotationsandprocedure},
    the assignments $(I_1,\ldots,I_n)$ can be realized as follows:  for each $i \in [n]$, let $I_i = \mathbb{1}[U_i<g\{4\Lambda_{i-1}(\bm{Z}_i)\}]$,  where $\{U_i: i \in [n]\}$ are independent uniform random variables on $(0,1)$ for each $i \in [n]$, and  are independent from $\{\bm{Z}_i: i \in [n]\}$. In this equivalent construction, $I_i$ is a function of $\{\bm{Z}_j,U_j: j \in [i]\}$ for $i \in [n]$. Thus it suffices to show that for any non-negative, measurable  function $\tilde{\zeta}:\mathcal{Z}^{n}\times (0,1)^n \to [0,\infty)$,  
    $$E_{\bm{p}}\{\tilde{\zeta}(\bm{Z_1,\ldots,Z_n},U_1,\ldots,U_n)\} =E_{\bm{p}_0}[\exp\{\ell_n(\bm{p};\bm{p}_0)\}\tilde{\zeta}(\bm{Z_1,\ldots,Z_n},U_1,\ldots,U_n)].$$

Due to the definition of $\ell_n(\bm{p};\bm{p}_0)$ in Equation \eqref{ell}, we have 
\begin{align*}
        &\mathbb{E}_{\bm{p}}\left\{\tilde{\zeta}(\bm{Z}_1,\ldots,\bm{Z}_n,U_1,\ldots,U_n)\right\}
        \\
        &= \int_{0}^{1} \cdots \int_{0}^{1} \sum_{\bm{z}_1 \in 
        \mathcal{Z}}\cdots \sum_{\bm{z}_n \in \mathcal{Z}} \left\{\tilde{\zeta}(\bm{z}_1,\ldots,\bm{z}_n,u_1,\ldots,u_n)\prod_{i=1}^{n} \bm{p}({\bm{z}_i}) \right\} du_1\cdots du_n
        \\
        &= \int_{0}^{1} \cdots \int_{0}^{1}  \sum_{\bm{z}_1 \in 
        \mathcal{Z}}\cdots \sum_{\bm{z}_n \in \mathcal{Z}}\left[\tilde{\zeta}(\bm{z}_1,\ldots,\bm{z}_n,u_1,\ldots,u_n)\left\{\prod_{i=1}^{n} \frac{ \bm{p}(\bm{z}_i)}{\bm{p}_0(\bm{z}_i)} \right\}\prod_{i=1}^{n} \bm{p}_0(\bm{z}_i) \right] du_1\cdots du_n
        \\
        &=\mathbb{E}_{\bm{p}_0}\left[\exp\left\{\ell_n(\bm{p};\bm{p}_0)\right\}\tilde{\zeta}(\bm{Z_1,\ldots,Z_n},U_1,\ldots,U_n)\right],
    \end{align*}
    which completes the proof.
\end{proof}

\subsection{Proofs of Lemma \ref{lemma:joint}} \label{app:proof_joint_local}

Before proving Lemma \ref{lemma:joint}, we state two lemmas.
Recall the Markov chain $(\bm{\Lambda}_n)_{n\geq 1}$ with the state space $\mathcal{S}$, and transition probabilities $\bm{P}_{\bm{\Lambda}}$ in Lemma \ref{lemma:relevant properties1}(i). Also recall for each $\bm{z} \in \mathcal{Z}$,  $\bar{g}_{\bm{z}}$ in \eqref{def:g_bar} and 
the unique Poisson equation solution $\hat{g}_{\bm{z}}$  in Lemma \ref{lemma:relevant properties2}(i). Further, for $i \geq 1$ and $\bm{z} \in \mathcal{Z}$, recall the sequence $\Delta M_{i,z}$  in Equation \eqref{deltaMiz} in Lemma \ref{lemma:relevant properties2}(iii).

For each $\bm{z} \in \mathcal{Z}$, define $\bm{B_z} = \bm{\Lambda}_i-\bm{\Lambda}_{i-1}$ if $S_i(\bm{z})-S_{i-1}(\bm{z})=1$; note that $\bm{B}_{\bm{z}}$ does not depend on $\bm{\Lambda}_{i-1}$ or $i$. For each $\bm{z},\bm{z'} \in \mathcal{Z}$, define two functions $H_{\bm{z},\bm{z}'}^{\pm}: \mathcal{S} \to \mathbb{R}$ as follows: for  $\bm{\lambda} \in \mathcal{S}$
\begin{align}
    \label{def:H_zzp}
 H_{\bm{z},\bm{z}'}^{\pm}(\bm{\lambda}) =     \hat{g}_{\bm{z}}(\bm{\lambda}+\bm{B}_{\bm{z}'})g\{4\bm{\lambda}(\bm{z}')\}\pm\hat{g}_{\bm{z}}(\bm{\lambda}-\bm{B}_{\bm{z}'})g\{-4\bm{\lambda}(\bm{z}')\}.
\end{align}
By Corollary \ref{cor:oddfunction}, $\hat{g}_{\bm{z}}(\cdot)$ is an odd function, and as a result, $ H_{\bm{z},\bm{z}'}^{+}(\cdot)$ is an odd function, while 
$H_{\bm{z},\bm{z}'}^{-}(\cdot)$ is an even function.
Note further that both $H_{\bm{z},\bm{z}'}^{\pm}(\cdot)$ are \textit{not} symmetric in $(\bm{z},\bm{z}')$.

\begin{mylem} \label{lemma:preface2}
Let $\bm{\delta} \in \mathbb{R}^m$ be any vector such that $\sum_{\bm{z} \in \mathcal{Z}} \bm{\delta}(\bm{z}) = 0$. For any $\bm{z} \in \mathcal{Z}$,  as $n \to \infty$,
    $$n^{-1} \sum_{i=1}^{n} \mathbb{E}\left\{\Delta M_{i,\bm{z}} {\bm{\delta}(\bm{Z}_i})/{\bm{p}_0(\bm{Z}_i)} \mid \mathcal{F}_{i-1}\right\}\to 0 \quad \text{in probability}.$$
\end{mylem}

\begin{mylem} \label{lemma:preface1}
Let $\bm{z},\bm{z'} \in \mathcal{Z}$ be fixed. As $n \to \infty$, in probability,
 $$ n^{-1}\sum_{i=1}^{n}\mathbb{E}\left(\Delta M_{i,\bm{z}} \Delta  M_{i,\bm{z}'} \mid  \mathcal{F}_{i-1}\right)\;\to\;
  \mathbb{1}(\bm{z}=\bm{z}')\bm{p}_0(\bm{z}) +
 \bm{p}_0(\bm{z})\bm{p}_0(\bm{z}') \pi (H_{\bm{z,z'}}^{-}+H_{\bm{z',z}}^{-}).$$
\end{mylem}

The detailed proofs for these two lemmas above are relegated after the proof of Lemma \ref{lemma:joint}.

\begin{proof}[Proof of Lemma \ref{lemma:joint}]
From the definition of $\Delta M_{i,\bm{z}}$ and by the telescoping sum, for each $\bm{z} \in \mathcal{Z}$, 
$$n^{-1/2}\bm{S}_n(\bm{z}) = n^{-1/2}\sum_{i=1}^{n}\Delta M_{i,\bm{z}}-n^{-1/2}\bm{p}_0(\bm{z})\{\hat{g}_{\bm{z}}(\bm{\Lambda}_n)-\hat{g}_{\bm{z}}(\bm{\Lambda}_0)\}.
$$
Since $\bm{\Lambda}_0 = \bm{0}$ and
by Lemma \ref{lemma:relevant properties2}(ii), we have $n^{-1/2}\left(\hat{g}_{\bm{z}}(\bm{\Lambda}_n)-\hat{g}_{\bm{z}}(\bm{\Lambda}_0)\right) = o_{\bm{p}_0}(1)$. Then in view of Lemma \ref{lemma:LAN}, it suffices to show
   \begin{equation*}
     n^{-1/2}\sum_{i=1}^{n}\begin{Bmatrix}   \Delta M_{i,\bm{z}}\\  {\bm{\delta}(\bm{Z}_i)}/{\bm{p}_0(\bm{Z}_i)}
      \end{Bmatrix}
      \;\; \Rightarrow_{\bm{p}_0}\;\; 
      \mathcal{N}_{m+1}\left(
      \begin{pmatrix}\bm{0}_m\\
        0
        \end{pmatrix}
      , \begin{pmatrix} \bm{\Sigma}_\infty(\bm{p}_0) &\bm{0}_m \\ \bm{0}^\text{T}_m & \tau^2\end{pmatrix} \right),
    \end{equation*}
where for $\bm{z},\bm{z}' \in \mathcal{Z}$, the $(\bm{z},\bm{z}')$-th entry of the matrix $\bm{\Sigma}_{\infty}(\bm{p}_0)$ is 
\begin{equation}
\label{def:app_sigma_infty}
\eta_{\bm{z},\bm{z}'}  =  \mathbb{1}(\bm{z}=\bm{z}')\bm{p}_0(\bm{z}) +
 \bm{p}_0(\bm{z})\bm{p}_0(\bm{z}') \pi (H_{\bm{z,z'}}^{-}+H_{\bm{z',z}}^{-}),
\end{equation}
and recall $\tau^2$ in Lemma \ref{lemma:joint}. 
By the Cramer-Wold device, it is equivalent to show that  for any  vector 
${\bm{a}} = \{\bm{a}_{\bm{z}}: \bm{z} \in \mathcal{Z}\}$ and scalar $b_{\ell}$,
\begin{align}
    \begin{split}
        &n^{-1/2}\sum_{i=1}^{n} \Delta \widetilde{M}_{i}(\bm{a},b_{\ell})
\;\; \Rightarrow_{\bm{p}_0}\;\; 
\mathcal{N}_1\left(0, {\bm{a}}^\text{T} \bm{\Sigma}_\infty(\bm{p}_0) {\bm{a}} + b_{\ell}^2 \tau^2 \right),
    \end{split} \label{aux:cw}
\end{align}
where $\Delta \widetilde{M}_{i}(\bm{a},b_{\ell}) = \left(\sum_{\bm{z} \in \mathcal{Z}} \bm{a}_{\bm{z}}  \Delta M_{i,\bm{z}} \right) + b_{\ell}  {\bm{\delta}(\bm{Z}_i)}/{\bm{p}_0(\bm{Z}_i)}$ for $i \in [n]$. 
Next, we fix some arbitrary ${\bm{a}} \in \mathbb{R}^{m}$ and $b_{\ell} \in \mathbb{R}$, write $\Delta \widetilde{M}_{i}$ for $\Delta \widetilde{M}_{i}(\bm{a},b_{\ell})$ and prove the above claim.

Note that for each $i \in [n]$,
\begin{align*}
    \mathbb{E}\{{\bm{\delta}(\bm{Z}_i)}/{\bm{p}_0(\bm{Z}_i)}\mid \mathcal{F}_{i-1}\}
    = \sum_{\bm{z} \in \mathcal{Z}} \bm{p}_0(\bm{z}) \times \bm{\delta}(\bm{z})/\bm{p}_0(\bm{z})  = 0,
\end{align*}
where the last equality is because $\bm{\delta} = \lim_{n} n^{1/2}(\bm{p}_n -\bm{p}_0)$. Thus by  Lemma \ref{lemma:relevant properties2}(iii), $\{\Delta \widetilde{M}_{i} : i \in [n]\}$ is a sequence of $\{\mathcal{F}_i: 0 \leq i \leq n\}$-martingale differences.

By the martingale central limit theorem \cite[Corollary 3.1.]{hall}, to show \eqref{aux:cw}, it suffices to verify the following Lindeberg's condition and the convergence of predictive variance: 
\begin{align}
        &\sum_{i=1}^{n} \mathbb{E}\left\{\left(n^{-1/2}\Delta \widetilde{M}_{i}\right)^2 \mathbb{1}\left(|n^{-1/2}\Delta \widetilde{M}_{i} \mid>\varepsilon\right)|\mathcal{F}_{i-1}\right\}\to 0 \quad \text{in probability}, \text{ for all} \; \varepsilon>0,
        \label{conLind}\\
        &\sum_{i=1}^{n} \mathbb{E}\left\{ \left(n^{-1/2}  \Delta\widetilde{M}_{i}\right)^2\mid\mathcal{F}_{i-1}\right\}\to \sum_{\bm{z},\bm{z}' \in \mathcal{Z}}
    \bm{a}_{\bm{z}}  \bm{a}_{\bm{z}'} \eta_{\bm{z},\bm{z}'}   +b_{\ell}^2 \tau^2 \quad \text{in probability}. \label{conVar}
\end{align}

We first consider condition \eqref{conLind}. 
For each $i \in [n]$, $\mathbb{E} \{\bm{\delta}(\bm{Z}_i)/\bm{p}_0(\bm{Z}_i)\}^4  = \sum_{\bm{z} \in \mathcal{Z}} \bm{\delta}^4(\bm{z})/\bm{p}_0^3(\bm{z})$, which implies that 
$n^{-2} \sum_{i=1}^{n} \mathbb{E}[ \{\bm{\delta}(\bm{Z}_i)/\bm{p}_0(\bm{Z}_i)\}^4 ] \to  0$. Further, 
we have for each $\bm{z} \in \mathcal{Z}$,
\begin{align*}
    0 \leq n^{-2} \sum_{i=1}^{n}\mathbb{E}\left( \Delta M_{i,\bm{z}}^4 \right)
    \leq n^{-1} \max_{i \in [n]}\mathbb{E}\left( \Delta M_{i,\bm{z}}^4 \right) 
    \; \to \; 0, 
\end{align*}
where the last convergence is due to Lemma \ref{lemma:relevant properties2}(ii) 
and the definition of $\Delta M_{i,\bm{z}}$ in \eqref{deltaMiz} (note that $|\bm{S}_i(\bm{z})-\bm{S}_{i-1}(\bm{z})| \leq 1$). Thus, we have $\sum_{i=1}^{n} \mathbb{E}\left\{(n^{-1/2}\Delta \widetilde{M}_{i})^4 \right\}$ converges to zero, which, due to 
Markov inequality, verifies condition \eqref{conLind}.

Next, we verify  condition \eqref{conVar}. Note the following decomposition
    $$
     \mathbb{E}\left(n^{-1} \Delta  \widetilde{M}_{i}^2\mid\mathcal{F}_{i-1}\right) = 
     \sum_{\bm{z},\bm{z}' \in \mathcal{Z}}
    \bm{a}_{\bm{z}}\bm{a}_{\bm{z}'} \bm{I}^{(i)}(\bm{z},\bm{z}') +b_{\ell}^2 \bm{II}^{(i)} + 2\sum_{\bm{z}\in \mathcal{Z}} 
     \bm{a}_{\bm{z}} b_{\ell}
     \bm{III}^{(i)}(\bm{z}),
     $$
     where we define
     \begin{align*}
     &\bm{I}^{(i)}(\bm{z},\bm{z}') = n^{-1} \mathbb{E}\left(\Delta M_{i,\bm{z}}\Delta M_{i,\bm{z}'} \mid  \mathcal{F}_{i-1}\right),  \quad 
     \bm{II}^{(i)} = n^{-1} \mathbb{E}[\left\{
     \bm{\delta}(\bm{Z}_i)/\bm{p}_0(\bm{Z}_i)
     \right\}^2 \mid \mathcal{F}_{i-1}], \\
     &  \bm{III}^{(i)}(\bm{z}) = n^{-1}\mathbb{E}\left\{\Delta M_{i,\bm{z}} \bm{\delta}(\bm{Z}_i)/\bm{p}_0(\bm{Z}_i) \mid  \mathcal{F}_{i-1}\right\}.
     \end{align*}
By Lemma \ref{lemma:preface1} and \ref{lemma:preface2}, for each $\bm{z},\bm{z}' \in \mathcal{Z}$, we have
\begin{align*}
    \sum_{i=1}^{n} \bm{I}^{(i)}(\bm{z},\bm{z}') \to  \eta_{\bm{z},\bm{z}'} \;\text{ in probability} \quad \text{and}\quad \sum_{i=1}^{n}  \bm{III}^{(i)}(\bm{z}) \to 0 \;\text{ in probability}.
\end{align*}

Finally, for the second term,
    \begin{align*}
        \sum_{i=1}^{n}\bm{II}^{(i)} &= n^{-1} \sum_{i=1}^{n}  \sum_{\bm{z}\in \mathcal{Z}}\frac{\bm{\delta}^2(\bm{z})}{\bm{p}_0(\bm{z})} =  \tau^2,
    \end{align*}
    which verifies condition \eqref{conVar}. The proof is complete.
\end{proof}

We finally turn to the proofs of Lemma \ref{lemma:preface2} and \ref{lemma:preface1}. 

\begin{proof}[Proof of Lemma \ref{lemma:preface2}]
Fix some $\bm{z} \in \mathcal{Z}$. By definition,  for each $\bm{z}' \in \mathcal{Z}$, given $\mathcal{F}_{i-1}$,
the conditional probabilities for the events 
$\left\{\bm{Z}_i=\bm{z}', I_i=1\right\}$ and 
$\left\{\bm{Z}_i=\bm{z}', I_i=0\right\}$ are respectively, 
$\bm{p}_0(\bm{z}')g\{4\bm{\Lambda}_{i-1}(\bm{z'})\}$ and
$\bm{p}_0(\bm{z}')g\{-4\bm{\Lambda}_{i-1}(\bm{z'})\}$.
On  the event $\left\{\bm{Z}_i=\bm{z}', I_i=1\right\}$, 
    $$
        \Delta M_{i,\bm{z}}\frac{\bm{\delta}(\bm{Z}_i)}{\bm{p}_0(\bm{Z}_i)}
        =\frac{\bm{\delta}(\bm{z}')}{\bm{p}_0(\bm{z}')}[\mathbb{1}(\bm{z}=\bm{z}')+\bm{p}_0(\bm{z})\{\hat{g}_{\bm{z}}(\bm{\Lambda}_{i-1}+\bm{B}_{\bm{z}'})-\hat{g}_{\bm{z}}(\bm{\Lambda}_{i-1})\}],
    $$
and on the event $\left\{\bm{Z}_i=\bm{z}', I_i=0\right\}$,  
    $$  
        \Delta M_{i,\bm{z}}\frac{\bm{\delta}(\bm{Z}_i)}{\bm{p}_0(\bm{Z}_i)}
        =\frac{\bm{\delta}(\bm{z}')}{\bm{p}_0(\bm{z}')}[-\mathbb{1}(\bm{z}=\bm{z}')+\bm{p}_0(\bm{z})\{\hat{g}_{\bm{z}}(\bm{\Lambda}_{ i-1}-\bm{B}_{\bm{z}'})-\hat{g}_{\bm{z}}(\bm{\Lambda}_{i-1})\}].
    $$
Thus, since 
$g(x) + g(-x) = 1$ for $x \in \mathbb{R}$, we have
\begin{align*}
 &\mathbb{E}\left\{\Delta M_{i,\bm{z}} {\bm{\delta}(\bm{Z}_i})/{\bm{p}_0(\bm{Z}_i)} \mid \mathcal{F}_{i-1}\right\}\\
 = &\bm{\delta}(\bm{z}) \bar{g}_{\bm{z}}(\bm{\Lambda}_{i-1}) 
 + \sum_{\bm{z}' \in \mathcal{Z}} \bm{p}_0(\bm{z})\bm{\delta}(\bm{z}') H_{\bm{z},\bm{z}'}^{+}(\bm{\Lambda}_{i-1})
 - \sum_{\bm{z}' \in \mathcal{Z}}  \bm{p}_0(\bm{z})\bm{\delta}(\bm{z}')\hat{g}_{\bm{z}}(\bm{\Lambda}_{i-1}),
\end{align*}
where $H_{\bm{z},\bm{z}'}^{+}$ is defined in Equation \eqref{def:H_zzp}. 
By assumption, $\sum_{\bm{z}' \in \mathcal{Z}}  \bm{p}_0(\bm{z})\bm{\delta}(\bm{z}')\hat{g}_{\bm{z}}(\bm{\Lambda}_{i-1}) = 0$.

By  Lemma \ref{lemma:relevant properties1} and
\cite{norris1998markov},  
there exits a unique invariant distribution $\pi$ on $\mathcal{S}$ for the Markov chain $(\bm{\Lambda}_n)_{n\geq 1}$, and for any odd,  integrable function $f: \mathcal{S} \to \mathbb{R}$, $\pi f = 0$. In particular, both $\bar{g}_{\bm{z}}(\cdot)$ and $H_{\bm{z},\bm{z}'}^{+}(\cdot)$ are odd, and thus $\pi \bar{g}_{\bm{z}} = 0$ and $\pi H_{\bm{z},\bm{z}'}^{+} = 0$. Then, by the Ergodic Theorem \cite[Theorem 1.10.2]{norris1998markov}, as $n \to \infty$, for each $\bm{z}' \in \mathcal{Z}$,
\begin{align*}
     &n^{-1}\sum_{i=1}^{n}\bm{\delta}(\bm{z}) \bar{g}_{\bm{z}}(\bm{\Lambda}_{i-1})\to \bm{\delta}(\bm{z})\pi \bar{g}_{\bm{z}}=0 \quad\text{in probability},\\
     & n^{-1}\sum_{i=1}^{n} \bm{p}_0(\bm{z})\bm{\delta}(\bm{z}') H_{\bm{z},\bm{z}'}^{+}(\bm{\Lambda}_{i-1}) \to \bm{p}_0(\bm{z})\bm{\delta}(\bm{z}') \pi H_{\bm{z},\bm{z}'}^{+}=0 \quad\text{in probability},
\end{align*}
which completes the proof.
\end{proof}

\begin{proof}[Proof of Lemma \ref{lemma:preface1}]
Fix $\bm{z},\bm{z}' \in \mathcal{Z}$. Denote by $\Delta \bm{S}_i(\bm{z}) = \bm{S}_i(\bm{z})-\bm{S}_{i-1}(\bm{z})$. Note the following decomposition for $\Delta  M_{i,\bm{z}} \Delta  M_{i,\bm{z}'}$:
\begin{align*}
 & \Delta \bm{S}_i(\bm{z})\Delta \bm{S}_i(\bm{z}')   +  \bm{p}_0(\bm{z})\bm{p}_0(\bm{z}')  \{\hat{g}_{\bm{z}}(\bm{\Lambda}_i)-\hat{g}_{\bm{z}}(\bm{\Lambda}_{i-1})\} \{\hat{g}_{\bm{z}'}(\bm{\Lambda}_i)-\hat{g}_{\bm{z}'}(\bm{\Lambda}_{i-1})\}
 \\
 &+ \Delta \bm{S}_i(\bm{z}) \bm{p}_0(\bm{z}')\{\hat{g}_{\bm{z}'}(\bm{\Lambda}_i)-\hat{g}_{\bm{z}'}(\bm{\Lambda}_{i-1})\}  +   \Delta \bm{S}_i(\bm{z}') \bm{p}_0(\bm{z})\{\hat{g}_{\bm{z}}(\bm{\Lambda}_i)-\hat{g}_{\bm{z}}(\bm{\Lambda}_{i-1})\}.
\end{align*}
By  similar calculations as in the proof of Lemma \ref{lemma:preface2}, we have
\begin{align*}
&\mathbb{E}\left\{\Delta \bm{S}_i(\bm{z})\Delta \bm{S}_i(\bm{z}') \mid \mathcal{F}_{i-1} \right\} = \mathbb{1}(\bm{z}=\bm{z}') \bm{p}_0(\bm{z}),\\
&\mathbb{E}\left[ \Delta \bm{S}_i(\bm{z}) \{\hat{g}_{\bm{z}'}(\bm{\Lambda}_i)-\hat{g}_{\bm{z}'}(\bm{\Lambda}_{i-1})\} \mid \mathcal{F}_{i-1}
    \right]
    = \bm{p}_0(\bm{z}) \left\{ H_{\bm{z}',\bm{z}}^{-}(\bm{\Lambda}_{i-1})-   \bar{g}_{\bm{z}}(\bm{\Lambda}_{i-1}) \hat{g}_{\bm{z}'}(\bm{\Lambda}_{i-1})\right\},\\
  &\mathbb{E}\left[  \Delta \bm{S}_i(\bm{z}') \{\hat{g}_{\bm{z}}(\bm{\Lambda}_i)-\hat{g}_{\bm{z}}(\bm{\Lambda}_{i-1})\} \mid \mathcal{F}_{i-1}
    \right] 
    =\bm{p}_0(\bm{z}') \left\{ H_{\bm{z},\bm{z}'}^{-}(\bm{\Lambda}_{i-1})-   \bar{g}_{\bm{z}'}(\bm{\Lambda_{i-1}}) \hat{g}_{\bm{z}}(\bm{\Lambda}_{i-1})\right\},\\
&\mathbb{E}\left[   \{\hat{g}_{\bm{z}}(\bm{\Lambda}_i)-\hat{g}_{\bm{z}}(\bm{\Lambda}_{i-1})\} \{\hat{g}_{\bm{z}'}(\bm{\Lambda}_i)-\hat{g}_{\bm{z}'}(\bm{\Lambda}_{i-1})\} \mid \mathcal{F}_{i-1}
    \right] 
    = 
    \mathbb{E}\left\{\hat{g}_{\bm{z}}(\bm{\Lambda}_i)\hat{g}_{\bm{z}'}(\bm{\Lambda}_i) \mid \mathcal{F}_{i-1}\right\} \\
&    - \{\bm{P}_{\bm{\Lambda}} \hat{g}_{\bm{z}}(\bm{\Lambda}_{i-1})\}
     \hat{g}_{\bm{z}'}(\bm{\Lambda}_{i-1})
     -
     \{\bm{P}_{\bm{\Lambda}} \hat{g}_{\bm{z}'}(\bm{\Lambda}_{i-1})\}
     \hat{g}_{\bm{z}}(\bm{\Lambda}_{i-1})
    + \hat{g}_{\bm{z}}(\bm{\Lambda}_{i-1})\hat{g}_{\bm{z}'}(\bm{\Lambda}_{i-1}).
\end{align*}

By the Poisson equation in  Lemma \ref{lemma:relevant properties2}(i), we have
\begin{align*}
    &\bm{P}_{\bm{\Lambda}} \hat{g}_{\bm{z}}(\bm{\Lambda}_{i-1}) = \hat{g}_{\bm{z}}(\bm{\Lambda}_{i-1}) - \bar{g}_{\bm{z}}(\bm{\Lambda}_{i-1}), \quad \bm{P}_{\bm{\Lambda}} \hat{g}_{\bm{z}'}(\bm{\Lambda}_{i-1}) = \hat{g}_{\bm{z}'}(\bm{\Lambda}_{i-1}) - \bar{g}_{\bm{z}'}(\bm{\Lambda}_{i-1}).
\end{align*}

Now putting the above calculations together, we have
\begin{align*}
    &\mathbb{E}\left( \Delta  M_{i,\bm{z}} \Delta  M_{i,\bm{z}'}\mid \mathcal{F}_{i-1}\right)
    = \mathbb{1}(\bm{z}=\bm{z}') \bm{p}_0(\bm{z}) +  \bm{p}_0(\bm{z})\bm{p}_0(\bm{z}')\{H_{\bm{z}',\bm{z}}^{-}(\bm{\Lambda}_{i-1})+H_{\bm{z},\bm{z}'}^{-}(\bm{\Lambda}_{i-1})\}\\
    +&\bm{p}_0(\bm{z})\bm{p}_0(\bm{z}')\left[ \mathbb{E}\left\{\hat{g}_{\bm{z}}(\bm{\Lambda}_i)\hat{g}_{\bm{z}'}(\bm{\Lambda}_i) \mid \mathcal{F}_{i-1}\right\}-
    \hat{g}_{\bm{z}}(\bm{\Lambda}_{i-1})\hat{g}_{\bm{z}'}(\bm{\Lambda}_{i-1})
    \right].
\end{align*}

Due to  Lemma \ref{lemma:relevant properties1},
by the Ergodic Theorem \cite[Theorem 1.10.2]{norris1998markov}, we have that as $n \to \infty$, 
\begin{align*}
    n^{-1} \sum_{i=1}^{n}\{H_{\bm{z}',\bm{z}}^{-}(\bm{\Lambda}_{i-1})+H_{\bm{z},\bm{z}'}^{-}(\bm{\Lambda}_{i-1})\} \to \pi (H_{\bm{z}',\bm{z}}^{-}+H_{\bm{z},\bm{z}'
    }^{-}), \quad \text{in probability}.
\end{align*}
Then the proof is complete due to Lemma \ref{lemma:aux_bnd_md}
below.
\end{proof}

\begin{mylem} \label{lemma:aux_bnd_md}
Let $\bm{z},\bm{z'} \in \mathcal{Z}$ be fixed. As $n \to \infty$,
\begin{align*}
    n^{-1} \sum_{i=1}^{n}\left[ \mathbb{E}\left\{\hat{g}_{\bm{z}}(\bm{\Lambda}_i)\hat{g}_{\bm{z}'}(\bm{\Lambda}_i) \mid \mathcal{F}_{i-1}\right\} -
    \hat{g}_{\bm{z}}(\bm{\Lambda}_{i-1})\hat{g}_{\bm{z}'}(\bm{\Lambda}_{i-1})
    \right] \;\; \to \;\; 0, \quad \text{in probability}.
\end{align*}
\end{mylem}
\begin{proof}
Denote by $\Delta G_i = \hat{g}_{\bm{z}}(\bm{\Lambda}_{i})\hat{g}_{\bm{z}'}(\bm{\Lambda}_{i}) - \mathbb{E}\left\{\hat{g}_{\bm{z}}(\bm{\Lambda}_i)\hat{g}_{\bm{z}'}(\bm{\Lambda}_i) \mid \mathcal{F}_{i-1}\right\}$
 for $i \geq 1$.
Then $\{\Delta G_i: i \geq 1\}$ is a $\{\mathcal{F}_{i}:i\geq 1\}$-martingale difference sequence. Note that
\begin{align*}
&n^{-1}\sum_{i=1}^{n}\left[ \mathbb{E}\left\{\hat{g}_{\bm{z}}(\bm{\Lambda}_i)\hat{g}_{\bm{z}'}(\bm{\Lambda}_i) \mid \mathcal{F}_{i-1}\right\} -\hat{g}_{\bm{z}}(\bm{\Lambda}_{i-1})\hat{g}_{\bm{z}'}(\bm{\Lambda}_{i-1})
    \right] \\
    =& 
n^{-1}\left[\mathbb{E}\left\{\hat{g}_{\bm{z}}(\bm{\Lambda}_n)\hat{g}_{\bm{z}'}(\bm{\Lambda}_{n})\mid\mathcal{F}_{n-1}\right\}-\hat{g}_{\bm{z}}(\bm{\Lambda}_0)\hat{g}_{\bm{z}'}(\bm{\Lambda}_0) \right] -
n^{-1}\sum_{i=1}^{n-1} \Delta G_i.
\end{align*}

First, by Jensen's inequality for conditional expectations, and due to Lemma \ref{lemma:relevant properties2}(ii),
\begin{align*}
&n^{-2}\mathbb{E}\left[\mathbb{E}\left\{\hat{g}_{\bm{z}}(\bm{\Lambda}_n)\hat{g}_{\bm{z}'}(\bm{\Lambda}_{n})\mid\mathcal{F}_{n-1}\right\}-\hat{g}_{\bm{z}}(\bm{\Lambda}_0)\hat{g}_{\bm{z}'}(\bm{\Lambda}_0)\right]^2 
\\
\leq &  n^{-2} \left( \mathbb{E}\left[\{\hat{g}_{\bm{z}}(\bm{\Lambda}_n)\}^4
 \right]+ 
 \mathbb{E}\left[\{\hat{g}_{\bm{z}'}(\bm{\Lambda}_n)\}^4
 \right]+
\mathbb{E}\left[\{\hat{g}_{\bm{z}}(\bm{\Lambda}_0)\}^4\right]+
\mathbb{E}\left[\{\hat{g}_{\bm{z}'}(\bm{\Lambda}_0)\}^4\right] \right)\;\to\; 0,
\end{align*}
which implies that $n^{-1}\left[\mathbb{E}\left\{\hat{g}_{\bm{z}}(\bm{\Lambda}_n)\hat{g}_{\bm{z}'}(\bm{\Lambda}_{n})\mid\mathcal{F}_{n-1}\right\}-\hat{g}_{\bm{z}}(\bm{\Lambda}_0)\hat{g}_{\bm{z}'}(\bm{\Lambda}_0) \right]$ converges to zero in probability.

Further, due to the property of martingale differences, 
\begin{align*}
    \mathbb{E}\left\{ \left( n^{-1}\sum_{i=1}^{n-1} \Delta G_i \right)^2\right\} = n^{-2} \sum_{i=1}^{n-1} \mathbb{E}\left\{\left( \Delta G_i \right)^2 \right\} \leq n^{-1} \max_{1 \leq i \leq  n} \mathbb{E}\left\{\left( \Delta G_i \right)^2 \right\}.
\end{align*}
By  similar arguments as above, $n^{-1}\max_{1 \leq i \leq  n} \mathbb{E}\left\{\left( \Delta G_i \right)^2 \right\}\to 0$, which implies that $n^{-1}\sum_{i=1}^{n-1} \Delta G_i$ converges to zero in probability. Then the proof is complete.
\end{proof}

\subsection{Proof of Lemma \ref{lemma:perturbed normality}}
\label{app:proof_general_local}
We first state and prove a lemma, which is a generalization of Lemma \ref{lemma:relevant properties1}(ii) to local alternatives. 
It is also needed in the proof of Theorem \ref{thm:bootstrap_consistency}.

\begin{mylem}\label{app:lemma_bnded_unif_moments}
Let $\bm{p}_n \in \bm{\Delta}_m$ be any sequence of deterministic  probability mass functions such that  $n^{1/2}(\bm{p}_n - \bm{p}_0)$ converges to some vector as $n \to \infty$. Then for any $r > 0$,
and $\bm{z} \in \mathcal{Z}$,
\begin{align*}
    \sup_{n \geq 1} \mathbb{E}_{\bm{p}_n} \left\{|n^{-1/2}\bm{S}_n(\bm{z})|^r\right\} < \infty.
\end{align*}
\end{mylem}
\begin{proof}
By Lemma \ref{lemma:changeofmeasure} and Cauchy-Schwartz inequality, 
\begin{align*}
        &\sup_n  \mathbb{E}_{\bm{p}_n}\left\{|n^{-1/2}\bm{S}_n(\bm{z})|^r\right\}\\ =& \sup_n E_{\bm{p}_0}\left[\exp\{\ell_n(\bm{p}_n;\bm{p}_0)\}|n^{-1/2}\bm{S}_n(\bm{z})|^r\right]
        \\
        \leq& \sup_{n} \left(\mathbb{E}_{\bm{p}_0}\left[\exp\{2\ell_n(\bm{p}_n;\bm{p}_0)\}\right]\right)^{1/2} \cdot\sup_n \left(\mathbb{E}_{\bm{p}_0}\left[\{n^{-1/2}\bm{S}_n(\bm{z})\}^{2r}\right]\right)^{1/2}.
\end{align*}
By Lemma \ref{lemma:relevant properties1}(ii), the second term above is finite. It remains to show that the first term is  finite.

Denote by $\bm{\delta}_n = n^{1/2}(\bm{p}_n-\bm{p}_0)$ for $n \geq 1$;
by definition,   $\sum_{\bm{z}\in\mathcal{Z}}\bm{\delta}_n(\bm{z}) = 0$. Further,
let 
$M =\sup_{n} \sup_{\bm{z} \in \mathcal{Z}}|\bm{\delta}_n(\bm{z})|$;
since $\{\bm{\delta}_n\}$ is convergent, we have $M < \infty$. Finally, due to independence and by definition, we have
\begin{align*}
    &\mathbb{E}_{\bm{p}_0}[\exp\{2\ell_n(\bm{p}_n;\bm{p}_0)\}]
    =\mathbb{E}_{\bm{p}_0}\left[\prod_{i=1}^{n}
    \left\{
    \frac{\bm{p}_0(\bm{Z}_i)+n^{-1/2}\bm{\delta}_n(\bm{Z}_i)}{\bm{p}_0(\bm{Z}_i)}\right\}^2\right]
    \\
    =&\prod_{i=1}^{n}\sum_{\bm{z}\in \mathcal{Z}}\frac{\{\bm{p}_0(\bm{z})+n^{-1/2}\bm{\delta}_n(\bm{z})\}^2}{\bm{p}_0(\bm{z})}
    =\prod_{i=1}^{n}\left\{1+\sum_{\bm{z}\in \mathcal{Z}}\frac{\bm{\delta}_n^2(\bm{z})}{n\bm{p}_0(\bm{z})}\right\} \leq \exp\left\{\sum_{\bm{z}\in \mathcal{Z}}\frac{M^2}{\bm{p}_0(\bm{z})}\right\}.
\end{align*}
Note that the above upper bound does not depend on $n$, and thus   
$\sup_{n} \mathbb{E}_{\bm{p}_0}\left[\exp\{2\ell_n(\bm{p}_n;\bm{p}_0)\}\right]$
is finite, which completes the proof.
\end{proof}

Given the above lemma and Lemma \ref{lemma:joint}, the proof of Lemma \ref{lemma:perturbed normality} is straightforward.

\begin{proof}[Proof of Lemma \ref{lemma:perturbed normality}]
By Lemma \ref{lemma:changeofmeasure}, the log-likelihood ratio of $\bm{p}_n$ against $\bm{p}_0$ for the \textit{joint} sequence $\{\bm{Z}_i, I_i: i \in [n]\}$ is $\ell_n(\bm{p}_n; \bm{p}_0)$. 
Note that $\bm{S}_n$ is a function of $\{\bm{Z}_i, I_i: i \in [n]\}$, and by Lemma \ref{lemma:joint}, $n^{-1/2} \bm{S}_n$
 and $\ell_n(\bm{p}_n; \bm{p}_0)$ converge jointly under $\bm{p}_0$ to a multivariate normal distribution. Due to the structure of the limiting distribution and by Le Cam's third lemma \cite[Theorem 6.6]{van},
we have 
$
n^{-1/2}\bm{S}_n\Rightarrow_{\bm{p}_n} \mathcal{N}_{m}(\bm{0},\bm{\Sigma}_\infty(\bm{p}_0))
$.

By Lemma \ref{app:lemma_bnded_unif_moments}, for each $\bm{z} \in \mathcal{Z}$, we have
$
\sup_n  \mathbb{E}_{\bm{p}_n}\left[\{n^{-1/2}\bm{S}_n(\bm{z})\}^4\right]<\infty
$. 
Then, if we
let $\bm{\Theta}$ be a random vector with the distribution  $\mathcal{N}_{m}(\bm{0},\bm{\Sigma}_\infty(\bm{p}_0))$, by the convergence of moments \cite[Theorem 2.20]{van}, we have as $n \to \infty$,
\begin{align*}
    \bm{\Sigma}_{n}(\bm{p}_n) = \mathbb{E}_{\bm{p}_n}\left\{(n^{-1/2}\bm{S}_n) (n^{-1/2} \bm{S}_n)^T\right\} \;\; \to \;\; \mathbb{E}\left( \bm{\Theta} \bm{\Theta}^T\right) = \bm{\Sigma}_{\infty}(\bm{p}_0),
\end{align*}
where  the first equality  is because  $\mathbb{E}_{\bm{p}_n}[\bm{S}_n] = \bm{0}_m$ due to symmetry. The proof is complete.
\end{proof}

\subsection{Proof of Theorem \ref{thm:bootstrap_consistency}}\label{app:proof_bootstrap_consistency}

We first state and prove a lemma, which may be viewed as the stochastic version of Lemma \ref{app:lemma_bnded_unif_moments}.

\begin{mylem}\label{app:lemma_bound_in_prob}
Let $\hat{\bm{p}}_n \in \bm{\Delta}_m$ be a sequence of estimators such that $n^{1/2}(\hat{\bm{p}}_n - \bm{p}_0)$ converges in distribution as $n \to \infty$. For any $r > 0$ and each $\bm{z} \in \mathcal{Z}$, as $n \to \infty$,
$\mathbb{E}_{\hat{\bm{p}}_n}\left\{|n^{-1/2}\bm{S}_n(\bm{z})|^r\right\}$ is bounded in probability.
\end{mylem}

\begin{proof}
Fix some $r > 0$ and $\bm{z} \in \mathcal{Z}$. Define $\zeta_{n}(\bm{p}) = \mathbb{E}_{\bm{p}}(|n^{-1/2}\bm{S}_n(\bm{z})|^r)$ 
for $n \geq 1$ and $\bm{p} \in \bm{\Delta}_m$.

Due to Lemma \ref{lemma:joint} and Lemma  \ref{app:lemma_bnded_unif_moments}, by similar arguments as for Lemma \ref{lemma:perturbed normality}, for any sequence of deterministic $\bm{p}_n \in \bm{\Delta}_m$ such that $n^{1/2}(\bm{p}_n-\bm{p}_0)$ is convergent, we have
$$
\zeta_{n}(\bm{p}_n) \to \mathbb{E}\left(|\mathcal{N}_{1}(0, \sigma_{\bm{z}}^2)|^r\right),
$$
where  $\mathcal{N}_{1}(0, \sigma_{\bm{z}}^2)$ denotes a random variable that has a normal distribution with mean zero and variance $\sigma_{\bm{z}}^2$, which is the $(\bm{z},\bm{z})$-th entry of $\bm{\Sigma}_{\infty}(\bm{p}_0)$.
Then by the same argument as for Theorem \ref{theorem:main}, we have 
$
\zeta_{n}(\hat{\bm{p}}_n) \to \mathbb{E}\left(|\mathcal{N}_{1}(0, \sigma_{\bm{z}}^2)|^r\right)$ in probability, 
which completes the proof.
\end{proof}

\begin{proof}[Proof of Theorem \ref{thm:bootstrap_consistency}]
For $\bm{z},\bm{z}' \in \mathcal{Z}$, denote by $\bm{\Sigma}_{n}(\hat{\bm{p}}_n; \bm{z},\bm{z}')$ the $(\bm{z},\bm{z}')$-th entry of $\bm{\Sigma}_n(\hat{\bm{p}}_n)$ and by $\hat{\bm{\mu}}^{B}(\hat{\bm{p}}_n;\bm{z})$ the $\bm{z}$-th entry of $\hat{\bm{\mu}}^{B}(\hat{\bm{p}}_n)$. Denote for each $b \in [B]$ and $n \geq 1$, $\tilde{\bm{S}}_n^{(b)} = n^{-1/2}{\bm{S}}_n^{(b)}$.

Note the following decomposition
$B^{-1}(B-1)\hat{\bm{\Sigma}}_n^{B}(\hat{\bm{p}}_n) - \bm{\Sigma}_{\infty}(\bm{p}_0) = \bm{I}_{n,B} +
\bm{II}_{n}$,
 where 
$$
\bm{I}_{n,B} = \frac{1}{B} \sum_{b=1}^{B} \left\{\tilde{\bm{S}}_n^{(b)} (\tilde{\bm{S}}_n^{(b)})^\text{T} - \bm{\Sigma}_n(\hat{\bm{p}}_n)\right\} - \hat{\bm{\mu}}^{B}(\hat{\bm{p}}_n)\{\hat{\bm{\mu}}^{B}(\hat{\bm{p}}_n)\}^\text{T},
$$
and
$\bm{II}_{n} = {\bm{\Sigma}}_n(\hat{\bm{p}}_n) - \bm{\Sigma}_{\infty}(\bm{p}_0)$. 
By Theorem \ref{theorem:main},  $\bm{II}_{n}$ converges to zero in probability as $n \to \infty$. Thus, it suffices to show that for each $\bm{z},\bm{z}' \in \mathcal{Z}$,
$\bm{I}_{n,B}(\bm{z},\bm{z}') \to 0 \text{ in probability}$. Fix some $\bm{z},\bm{z}' \in \mathcal{Z}$, and arbitrary $\eta > 0$.

Since $\mathbb{P}\{|\bm{I}_{n,B}(\bm{z},\bm{z}')| \geq \eta\} =\mathbb{E}\left[ \mathbb{P}\left\{|\bm{I}_{n,B}(\bm{z},\bm{z}')| \geq \eta \mid \mathcal{F}_n\right\}\right]
 $, by dominated convergence theorem, it suffices to show that $\mathbb{P}\left[|\bm{I}_{n,B}(\bm{z},\bm{z}')| \geq \eta \mid \mathcal{F}_n\right] \to 0$ in probability. Then, by Markov inequality, it suffices to show that $\mathbb{E}\left[\{\bm{I}_{n,B}(\bm{z},\bm{z}')\}^2   \mid \mathcal{F}_n\right] \to 0$ in probability.

Observe that conditional on $\mathcal{F}_n$ (in particular, on $\hat{\bm{p}}_n$), $\{\tilde{\bm{S}}_n^{(b)}: b \in [B]\}$ are independent and identically distributed with $\mathbb{E}\left\{\tilde{\bm{S}}_n^{(1)}(\bm{z})\tilde{\bm{S}}_n^{(1)}(\bm{z}')\mid \mathcal{F}_{n}\right\}=\bm{\Sigma}_{n}(\hat{\bm{p}}_n; \bm{z},\bm{z}')$. Using $(a+b)^2 \leq 2a^2 + 2b^2$ and 
$2ab \leq a^2+b^2$, we have
\begin{align*}
    &\mathbb{E}\left[\{\bm{I}_{n,B}(\bm{z},\bm{z}')\}^2   \mid \mathcal{F}_n\right] \\ \leq& 
    2  B^{-1}\mathbb{E}\left[ \{\tilde{\bm{S}}_n^{(1)}(\bm{z})\tilde{\bm{S}}_n^{(1)}(\bm{z}')\}^2  \mid \mathcal{F}_n \right] + 
    2\mathbb{E}\left[\left\{\hat{\bm{\mu}}^{B}(\hat{\bm{p}}_n;\bm{z}) \hat{\bm{\mu}}^{B}(\hat{\bm{p}}_n;\bm{z}')\right\}^2  \mid \mathcal{F}_n\right] \\
    \leq& \;\; 
    B^{-1} \mathbb{E}\left[ \left\{\tilde{\bm{S}}_n^{(1)}(\bm{z})\right\}^4 + \left\{\tilde{\bm{S}}_n^{(1)}(\bm{z}')\right\}^4 \mid \mathcal{F}_n \right]
    + \mathbb{E}\left[\left\{\hat{\bm{\mu}}^{B}(\hat{\bm{p}}_n;\bm{z})\right\}^4   + \left\{\hat{\bm{\mu}}^{B}(\hat{\bm{p}}_n;\bm{z}')\right\}^4  \mid \mathcal{F}_n\right].
\end{align*}

Due to symmetry, $\mathbb{E}\{\tilde{\bm{S}}_n^{(1)}(\bm{z}) \mid \mathcal{F}_{n}\} = 0$, and thus 
 \begin{align*}
     &\mathbb{E}\left[\left\{\hat{\bm{\mu}}^{B}(\hat{\bm{p}}_n;\bm{z})\right\}^4 \mid \mathcal{F}_n\right] \\=& B^{-3}  \mathbb{E}\left[ \left\{\tilde{\bm{S}}_n^{(1)}(\bm{z})\right\}^4 \mid \mathcal{F}_n \right] + 3B^{-3}(B-1)  \mathbb{E}\left[ \left\{\tilde{\bm{S}}_n^{(1)}(\bm{z})\tilde{\bm{S}}_n^{(2)}(\bm{z})\right\}^2 \mid \mathcal{F}_n \right]\\
    \leq& \;\;  4B^{-2}\mathbb{E}\left[ \left\{\tilde{\bm{S}}_n^{(1)}(\bm{z})\right\}^4 \mid \mathcal{F}_n \right] = 4B^{-2}\mathbb{E}_{\hat{\bm{p}}_n}\left[ \{{n^{-1/2}\bm{S}}_n(\bm{z})\}^4\right].
\end{align*}
Combining the above calculations, we have
\begin{align*}
    \mathbb{E}\left[\{\bm{I}_{n,B}(\bm{z},\bm{z}')\}^2   \mid \mathcal{F}_n\right] \leq 5B^{-1}\mathbb{E}_{\hat{\bm{p}}_n}\left[ \{n^{-1/2}\bm{S}_n(\bm{z})\}^4\right]+
    5B^{-1}\mathbb{E}_{\hat{\bm{p}}_n}\left[ \{n^{-1/2}\bm{S}_n(\bm{z}')\}^4\right].
\end{align*}
By Lemma \ref{app:lemma_bound_in_prob}, 
both $\mathbb{E}_{\hat{\bm{p}}_n}\left[ \{n^{-1/2}\bm{S}_n(\bm{z})\}^4\right]$ and 
$\mathbb{E}_{\hat{\bm{p}}_n}\left[ \{n^{-1/2}\bm{S}_n(\bm{z}')\}^4\right]$
are bounded in probability, 
which implies that as $\min\{n,B\} \to \infty$, $\mathbb{E}\left[\{\bm{I}_{n,B}(\bm{z},\bm{z}')\}^2   \mid \mathcal{F}_n\right] \to 0$ in probability. The proof is complete.
\end{proof}

\subsection{Justification for \eqref{def:specialcov}} \label{app:matrix}

In this subsection, we provide a proof for \eqref{def:specialcov}, that is, when there are two stratification factors each with two levels, i.e., $m = 2 \times 2$,  $\bm{\Sigma}_\infty(\bm{p}_0)$ has a form given in \eqref{def:specialcov}, for a general probabilistic mass function $\bm{p}_0$ and function $g(\cdot)$ in Step 3 of the minimization method in Subsection \ref{subsec:designnotationsandprocedure}.

When $m = 2 \times 2$, the strata in $\mathcal{Z}$ can be enumerated as
$\bm{z}_1=(0,0),\bm{z}_2=(0,1),\bm{z}_3=(1,0),\bm{z}_4=(1,1)$. By \citet[Corollary 3.2(b)]{hu}, the marginal balances are bounded in probability, that is,
\begin{align*}
    &\bm{S}_n(\bm{z}_1) + \bm{S}_n(\bm{z}_2) = O_{\bm{p}_0}(1), \quad
    \bm{S}_n(\bm{z}_1) + \bm{S}_n(\bm{z}_3) = O_{\bm{p}_0}(1),\\
    &\bm{S}_n(\bm{z}_2) + \bm{S}_n(\bm{z}_4) = O_{\bm{p}_0}(1), \quad
    \bm{S}_n(\bm{z}_3) + \bm{S}_n(\bm{z}_4) = O_{\bm{p}_0}(1).
\end{align*}
which implies that
\begin{align*}
    &n^{-1/2}\bm{S}_n(\bm{z}_2) = -n^{-1/2}\bm{S}_n(\bm{z}_1) + o_{\bm{p}_0}(1), \quad
    n^{-1/2}\bm{S}_n(\bm{z}_3) = -n^{-1/2}\bm{S}_n(\bm{z}_1) + o_{\bm{p}_0}(1),\\
&n^{-1/2}\bm{S}_n(\bm{z}_4) =  n^{-1/2}\bm{S}_n(\bm{z}_1) + o_{\bm{p}_0}(1).
\end{align*}
That is, $n^{-1/2}\bm{S}_n = n^{-1/2} \bm{S}_n(\bm{z}_1)(1,-1,-1,1)^\text{T} + o_{\bm{p}_0}(1)$, which completes the proof.

\section{Notations related to survival model} \label{app survival}
In this section, we present the expressions for $\psi$ and $\bm{G}$ in \eqref{def:R_size} for readers' convenience, which are introduced in \cite{ye}. Recall the setup and notations in Section \ref{sec:survival}, and
define random variables $O_{ij}$ for $i \in [n]$ and $j=0,1$ as follows
\begin{align*}
    O_{ij} &= \int_{0}^{\infty} \left\{\mu(t)^{1-I_i}-I_i\mu(t)\right\}\left\{dN_{ij}(t)-p(t)Y_{ij}(t)\exp\left(\bm{\beta}_*^\text{T}\bm{W}_i\right)dt\right\},
\end{align*}
 where $ \mu(t)=E\{I_i\mid Y_i(t)=1\}$, $p(t) = \mathbb{E}\{Y_i(t)h(t\mid\bm{V}_i,I_i)\}/\mathbb{E}\left\{Y_i(t)\exp\left(\bm{\beta_*^\text{T}W_i}\right)\right\}$, and $\bm{\beta}_*$ is the in-probability limit of 
 the maximum likelihood estimators $\hat{\bm{\beta}}_n$ under $H_0$. Then 
 $\psi = \mathbb{E}\{\text{var}(O_{i1}\mid\bm{Z}_i)+\text{var}(O_{i0}\mid\bm{Z}_i)\}/2$ and $\bm{G} = \{\mathbb{E}(O_{i1} \mid \bm{Z}_i = \bm{z}); \bm{z} \in \mathcal{Z}\}$.

\end{appendix}

\end{document}